\documentclass[manuscript,natbib=false,nonacm]{acmart}
\AtBeginDocument{%
  }

\newtheorem{remark}{Remark}
\usepackage[bibliography=common]{apxproof}

\AtEndPreamble{%
\theoremstyle{acmplain}
\newtheoremrep{theorem}{Theorem}[section]
\newtheoremrep{proposition}[theorem]{Proposition}
\newtheoremrep{lemma}[theorem]{Lemma}
\newtheoremrep{corollary}[theorem]{Corollary}
\newtheoremrep{remark}[theorem]{Remark}
\newtheoremrep{notation}[theorem]{Notation}
\theoremstyle{acmdefinition}
\newtheoremrep{definition}[theorem]{Definition}
\newtheoremrep{example}[theorem]{Example}
}

\usepackage{macros}

\usepackage{booktabs}
\usepackage{threeparttable}

\usepackage{lineno}
\usepackage{amsfonts}
\usepackage{bussproofs}
\usepackage{amsmath}
\usepackage{xcolor}
\usepackage{multirow}
\usepackage{makecell}
\usepackage{lscape}
\usepackage{stmaryrd}
\usepackage{graphicx}
\usepackage{comment}
\usepackage{mdframed}
\usepackage{enumerate}
\usepackage{float}

\usepackage{tikz}

\newcommand{\UniMultiset}[1]{[#1]}
\newcommand{\AntOmSetA}{\Omega}
\newcommand{\AntVecA}{\mathbf{x}}
\newcommand{\AntVecB}{\mathbf{y}}

\usepackage{hyperref}
\usepackage{enumitem}

\RequirePackage[
  datamodel=acmdatamodel,
  style=acmnumeric
  ]{biblatex}

\begin{document}

\title{Hypersequent Calculi Have Ackermannian Complexity}


\author{A. R. Balasubramanian}
\email{bayikudi@mpi-sws.org}
\orcid{0000-0002-7258-5445}
\affiliation{%
  \institution{Max Planck Institute for Software Systems (MPI-SWS)}
  \city{Saarbrücken}
  \country{Germany}
}

\author{Vitor Greati}
\email{v.rodrigues.greati@rug.nl}
\orcid{1234-5678-9012}
\author{Revantha Ramanayake}
\email{d.r.s.ramanayke@rug.nl}
\orcid{0000-0002-7940-9065}
\affiliation{%
  \institution{Bernoulli Institute, University of Groningen}
  \city{Groningen}
  \country{The Netherlands}
}

\renewcommand{\shortauthors}{Balasubramanian, Greati and Ramanayake}

\begin{abstract}
For substructural logics with contraction or weakening admitting cut-free sequent calculi, proof search was analyzed using well-quasi-orders on $\mathbb{N}^d$ (Dickson’s lemma), yielding Ackermannian upper bounds via controlled bad-sequence arguments. For hypersequent calculi, that argument lifted the ordering to the powerset, since a hypersequent is a (multi)set of sequents. This induces a jump from Ackermannian to hyper-Ackermannian complexity in the fast-growing hierarchy, suggesting that cut-free hypersequent calculi for extensions of the commutative Full Lambek calculus with contraction or weakening ($\mathbf{FL_{ec}}$/$\mathbf{FL_{ew}}$) inherently entail hyper-Ackermannian upper bounds.
We show that this intuition does not hold: every extension of $\mathbf{FL_{ec}}$ and $\mathbf{FL_{ew}}$ admitting a cut-free hypersequent calculus has an Ackermannian upper bound on provability.

To avoid the powerset, we exploit novel dependencies between individual sequents within any hypersequent in backward proof search. The weakening case, in particular, introduces a Karp-Miller style acceleration, and it improves the upper bound for the fundamental fuzzy logic $\mathbf{MTL}$. Our Ackermannian upper bound is optimal
for the contraction case (realized by the logic $\mathbf{FL_{ec}}$).
\end{abstract}



\received{22 February 2026}

\maketitle


\section{Introduction}
\label{sec:introduction}
Substructural logics arise by omitting or modifying structural rules such as contraction, weakening, exchange, or associativity from classical or intuitionistic proof calculi. This yields a broad family of systems---including linear logic, relevance logics, mathematical fuzzy logics, and the full Lambek calculus~\cite{Pa, Restall, GalJipKowOno07, MPT}. These systems are able to reasonable about resources in a broad sense, since assumptions can no longer be duplicated or discarded freely. 

The basic substructural logics are Full Lambek calculus $\mathbf{FL}$ extended with the structural rules of exchange, weakening, or contraction, denoted $\mathbf{FL_S}$ ($\mathrm{S}\subseteq\{e,w,c\}$). These have a cut-free sequent calculus, which implies a significant restriction on the proof search space, and this enables the proof-theoretic study of computational complexity for the logic. For provability---the topic of interest in this work---all the basic substructural logics excepting $\mathbf{FL}_{\mathbf{ec}}$ are PSPACE-complete. Meanwhile, $\mathbf{FL}_{\mathbf{ec}}$ was famously shown by Kripke~\cite{Kri59} in 1959 to be decidable, and Urquhart~\cite{urquhart1999} established its non-primitive complexity; indeed, it is complete for the class $\mathbf{F}_\omega$ in Schmitz's~\cite{schmitz2016hierar} ordinal-indexed fast-growing complexity hierarchy. Informally, a problem in $\mathbf{F}_\omega$ has complexity bounded by a single Ackermannian-type function composed with primitive recursive functions.

To understand the source of the $\mathbf{F}_\omega$ upper bound, observe that proof search~$P$ in $\mathbf{FL}_{\mathbf{ec}}$ can be restricted to branches that are \emph{bad sequences} under what is essentially the product ordering on $\mathbb{N}^d$. Dickson's lemma~\cite{dickson1913} states that this is a well-quasi-ordering (wqo). I.e., every branch in~$P$ is a sequence $(s_i)_{i\in\mathbb{N}}$ of sequents with the property that there is no $i<j$ such that $s_i\leq s_j$ (read as $s_j$ is contractible to $s_i$). The point of the wqo is that every bad sequence wrt it is finite. Hence we can do proof search in $\mathbf{FL}_{\mathbf{ec}}$ and termination is guaranteed, and decidability follows. Successive jumps in this bad sequence are controlled in the sense that there is a primitive recursive function~$g$ (depending only on proof calculus) and a norm $\UniNorm{\cdot}{}$ on sequents such that $\UniNorm{s_i}{}<g^i(\UniNorm{s_0}{})$. Every such controlled bad sequence has a maximum length that is given by a \emph{length theorem}; in the case of $\leq$, this is an Ackermanian-type function in $\UniNorm{s_0}{}$ (Figueira et al~\cite{figueira2011}). The claimed upper bound follows. 

Most extensions of $\mathbf{FL}_{\mathbf{ec}}$ do not admit cut-free sequent calculi, but infinitely many do in the generalized setting of cut-free hypersequent calculi~\cite{CiaGalTer08,CiaGalTer17}; specifically, axioms from every level of the substructural hierarchy~\cite{CiaGalTer08,Jer16} excepting the final level $\mathcal{N}_3$ which contains axioms like distributivity that often lead to undecidability in the substructural setting. A \emph{hypersequent}~\cite{Min68,Pot83,Avr87} is a multiset of sequents, written as $s_0 \VL \ldots \VL s_k$. Hypersequent rules properly generalize sequent rules by allowing multiple components in the conclusion to be utilized simultaneously. Ciabattoni et al~\cite{CiaGalTer08,CiaGalTer17} introduced a general theory to compute analytic hypersequent calculi for non-classical and substructural logics, including intermediate, many-valued, and fuzzy logics.
Balasubramanian et al.~\cite{BalLanRam21LICS} extended the complexity analysis of $\mathbf{FL_{ec}}$ to these hypersequent calculi, by lifting the wqo on $\mathbb{N}^d$ to the minoring orderings on $\mathcal{P}_f(\mathbb{N}^d)$. Since the length theorem on such powerset orderings is hyper-Ackermannian (i.e., in $\mathbf{F}_{\omega^\omega}$) (Balasubramanian~\cite{bala2020}), this showed that every axiomatic extension of $\mathbf{FL_{ec}}$ that has a cut-free hypersequent calculus has hyper-Ackermannian complexity. The same upper bound is obtained in \cite{BalLanRam21LICS} for extensions of $\mathbf{FL_{ew}}$; this time, the majoring wqo was coupled with forward proof search. 
That yielded the first upper bound for the fundamental fuzzy logic 
$\mathbf{MTL}$~\cite{EstGod01}---a longstanding open question~\cite{Han2017}---which is axiomatized over $\mathbf{FL_{ew}}$ by the prelinearity axiom $(p\imp q)\lor (q\imp p)$.

The above line of reasoning suggested that the jump from sequents to hypersequents (necessitated by the need for a cut-free proof calculus) also entails a jump to~$\mathbf{F}_{\omega^\omega}$. In this work, we show that this intuition does \emph{not} hold: we prove that every extension of $\mathbf{FL_{ec}}$ and $\mathbf{FL_{ew}}$ admitting a cut-free hypersequent calculus has an Ackermannian upper bound on provability. The key technical insight is avoiding a lift to a powerset ordering. Our solution is to exploit novel dependencies between components in a hypersequent, so every hypersequent---now viewed as a sequence with the sequent components placed in the order of being added---in the backward proof search is itself a controlled bad sequence (in the sequent ordering!). 
For the case of contraction, our bound is optimal in the sense that $\mathbf{FL_{ec}}$ already realizes it.

The weakening case is considerably more complicated. 
Vanilla backward proof search will not terminate so we extend it with a Karp–Miller style acceleration~\cite{KarpM69}.
Due to the interaction of multiple sequents within a hypersequent, the precise definition here is much more sophisticated than the accelerations used for other problems in automata theory, such as Petri nets coverability or in counter automata~\cite{KarpM69,LerouxS05,BardinFLP08}. We then couple this acceleration with the identification of a proof-theoretic gadget to implement a very specific form of contraction.
An immediate corollary lowers the upper bound for $\mathbf{MTL}$ to Ackermannian. That also applies to other (standard complete) fuzzy logics in the literature e.g., $n$-contractive extensions~\cite{CiaEstGod02,HorNogPet07} and the wmn-extensions~\cite{BalCiaSpe12} of $\mathbf{MTL}$.

Adding structure to bad sequences via proof calculi is of independent relevance to the infinite-state system community since complexity based on wqos may be improved by a finer analysis on the bad sequences (see e.g., Lazic and Schmitz~\cite{LazicS16,LazicS21}). Specifically, the proof calculi are variations of alternating branching vector addition systems with states (BVASS), where a configuration along a run can interact also with configurations visited before.

\section{Preliminaries}
\label{sec:preliminaries}

\subsection{Proof theory of substructural logics via hypersequents}
\label{sec:prelims-proof-theory}

Let $\UniPropVars$ be a countably infinite set of \emph{propositional variables}. 
The set $L$ of  \emph{formulas} is defined using the following grammar.
\[
\UniFmA:=\UniPropA\in \UniPropVars 
\VL 1 \VL 0 \VL \UniFmA\land \UniFmA \VL \UniFmA\lor \UniFmA \VL \UniFmA\fus \UniFmA \VL \UniFmA\imp \UniFmA
\]
A \emph{logic} (over $L$)
is a consequence relation $\vdash$ over $L$.
Given a logic $\vdash$,
the \emph{provability}
problem refers to
deciding
whether, given $\UniFmA \in L$, it is the case that
$\varnothing \vdash \UniFmA$.
A \emph{sequent} is an expression of the form $\UniSequent{\UniMSetFmA}{\UniMSetSucA}$, where~$\UniMSetFmA$ (the \emph{antecedent}) is a finite multiset of formulas and~$\UniMSetSucA$ (the \emph{succedent}) is a list that is either empty or consists of a single formula (we call it a \emph{succedent-list} or \emph{stoup}). We use $\UniMSetFmA, \UniMSetFmB, \UniMSetFmC, \UniMSetSucA$ to denote finite multisets of formulas.
The \emph{multiplicity} of a formula $\UniFmA$ 
in $\UniMSetFmA$ is the number of copies of $\UniFmA$ in $\UniMSetFmA$.
We will denote an enumeration
of the elements in $\UniMSetFmA$ (including repetitions) using $\UniMultiset{\UniFmA_1,\ldots,\UniFmA_k}$.
The \emph{support} of any multiset
is the set of its elements with non-zero multiplicity.



A \textit{hypersequent} is a (possibly empty) finite multiset of sequents. It is written explicitly as follows, for $k\geq 0$,
\begin{equation}
\label{hypersequent}
\UniSequent{\UniMSetFmA_{1}}{  \UniMSetSucA_{1}} \VL \ldots \VL \UniSequent{\UniMSetFmA_{k}}{  \UniMSetSucA_{k}}
\end{equation}
Each sequent~$\UniSequent{\UniMSetFmA_{i}}{\UniMSetSucA_{i}}$ is called a \textit{component} of the hypersequent.
For a hypersequent~$h$ and sequent~$s$, we write $s \in h$ to mean that $s$ is a component of $h$.
Given a finite
    set  $\UniSubfmlaHyperseqSet$ of formulas, an \emph{$\UniSubfmlaHyperseqSet$-(hyper)sequent}
    is a (hyper)sequent whose  formulas are all in $\UniSubfmlaHyperseqSet$.
A hypersequent containing no repetitions is called \emph{irredundant}.

A \emph{hypersequent calculus} is a set of \textit{hypersequent rules}.
Each rule has a \textit{conclusion} hypersequent and some number of \textit{premise} hypersequent(s).
All the calculi in this paper have the property of having only finitely many rules with the same conclusion.

Commonly, rules are presented
as instances of  \emph{rule schemas} (see Example~\ref{ex:rule-instances}).
 Rule schemas and rules are usually presented in fraction notation with the premises listed above the fraction line and the conclusion below the line (see Figure~\ref{figure-HFLec}).
A rule (schema) with no premises is called an \textit{initial hypersequent}
or \textit{axiomatic} rule.
We will often see a hypersequent calculus as a collection of
rule schemas, meaning
that its rules are all the possibles instantiations of these schemas, also referred to as \emph{rule instances}. 

In order to properly define rule schemas we need to consider variables of different types, depending on whether they will be instantiated by formulas, multisets of formulas, sequents, or hypersequents.
The conclusion and every premise of a rule schema has the form $\UniHyperSequentCompA_{1}|\ldots|\UniHyperSequentCompA_{k}$, where each~$\UniHyperSequentCompA_{i}$ is (i)~a hypersequent-variable (denoted by~$\UniHyperMSetA$), or (ii)~a sequent-variable (denoted by $s$), or (iii)~$\UniSequent{\UniMSetVariableListA}{\UniMSetSucA}$, or (iv)~$\UniSequent{\UniMSetVariableListA}{\phantom{\UniMSetSucA}}$.
Here $\UniMSetVariableListA$ is a finite list of multiset-variables (we use $\UniMSetFmA,\UniMSetFmB,\UniMSetFmC$ to denote multiset-variables) and formulas built from formula-variables $\UniFmA,\UniFmB,\UniFmC$ and proposition-variables $\UniSchPropA,\UniSchPropB,\UniSchPropC$; also, $\UniMSetSucA$ is a  succedent-list-variable.
The variables in the rule schema are collectively referred to as \textit{schematic variables}.
We refer to $\UniHyperSequentCompA_{1}|\ldots|\UniHyperSequentCompA_{k}$ as a \emph{hypersequent} and to each~$\UniHyperSequentCompA_{i}$ as a \textit{component} even though they are built from schematic variables, unlike in (\ref{hypersequent}) where they are built from formulas; this overloading of terminology is standard.

A \textit{rule instance} is obtained from a rule schema by uniformly instantiating the schematic variables with concrete objects of the corresponding type: a hypersequent-variable with a hypersequent, a sequent-variable with a sequent, a list-variable with a list of formulas, and so on. A succedent-list-variable is instantiated by a \emph{stoup} (i.e., a list that is either empty or consists of a single formula).

A component in the premise (resp., conclusion) of a rule schema that is not a hypersequent-variable is called an \textit{active component} (resp., \emph{principal component}).
We use the same terminology for the corresponding component in a rule instance.

\begin{example}
\label{ex:rule-instances}
We consider the rule schema~($\land$R) given by
\begin{center}
\begin{small}
\AxiomC{$\UniHyperMSetA \VL  \UniActiveCompHighlight{\UniSequent{\UniMSetFmA}{ \UniFmA}}$}
\AxiomC{$\UniHyperMSetA \VL  \UniActiveCompHighlight{\UniSequent{\UniMSetFmA}{ \UniFmB}}$}
\RightLabel{($\land$R)}
\BinaryInfC{$\UniHyperMSetA \VL  \UniActiveCompHighlight{\UniSequent{\UniMSetFmA}{ \UniFmA\land\UniFmB}}$}
\DisplayProof
\end{small}
\end{center}
where the active 
components (in the premises) and principal components (in the conclusion) have been highlighted with a surrounding box. Its premises are $\UniHyperMSetA \VL  \UniSequent{\UniMSetFmA}{ \UniFmA}$ and $\UniHyperMSetA \VL \UniSequent{\UniMSetFmA}{ \UniFmB}$, and its conclusion is~$\UniHyperMSetA \VL  \UniSequent{\UniMSetFmA}{ \UniFmA\land\UniFmB}$.
Below we give some rule instances of ($\land$R), where their active components have been highlighted.

\begin{center}
\begin{footnotesize}
\begin{tabular}{c@{\hspace{1em}}c}
\AxiomC{$\UniActiveCompHighlight{\UniSequent{}{\UniPropA}}$}
\AxiomC{$\UniActiveCompHighlight{\UniSequent{}{\UniPropB}}$}
\BinaryInfC{$\UniActiveCompHighlight{\UniSequent{}{\UniPropA\land\UniPropB}}$}
\DisplayProof
&
\AxiomC{$\UniSequent{\UniPropA}{\UniPropA} \VL \UniActiveCompHighlight{
\UniSequent{\UniPropC,\UniPropC}{\UniPropA\land \UniPropB}
}$}
\AxiomC{$\UniSequent{\UniPropA}{\UniPropA} \VL \UniActiveCompHighlight{\UniSequent{\UniPropC,\UniPropC}{\UniPropB}}$}
\BinaryInfC{$\UniSequent{\UniPropA}{\UniPropA} \VL \UniActiveCompHighlight{
    \UniSequent{\UniPropC,\UniPropC}{(\UniPropA\land \UniPropB)\land \UniPropB}
}$}
\DisplayProof
\\[3em]
\multicolumn{2}{c}{
\AxiomC{$\UniSequent{}{\UniPropA\fus\UniPropA} \VL \UniSequent{\UniPropB\imp \UniPropA}{\UniPropA} \VL \UniActiveCompHighlight{\UniSequent{\UniPropC}{ \UniPropA\land \UniPropB}}$}
\AxiomC{$\UniSequent{}{\UniPropA\fus\UniPropA} \VL \UniSequent{\UniPropB\imp \UniPropA}{\UniPropA} \VL \UniActiveCompHighlight{\UniSequent{\UniPropC}{\UniPropB}}$}
\BinaryInfC{$\UniSequent{}{\UniPropA\fus\UniPropA} \VL \UniSequent{\UniPropB\imp \UniPropA}{\UniPropA} \VL \UniActiveCompHighlight{\UniSequent{\UniPropC}{(\UniPropA\land \UniPropB)\land \UniPropB}}$}
\DisplayProof
}
\end{tabular}
\end{footnotesize}
\end{center}
\end{example}

The basic hypersequent calculus in this paper is $\UniFLeExtHCalc{}$ (see Fig~\ref{figure-HFLec}).
We omitted the (cut) rule, as cut elimination holds for this calculus.
$\UniFLeExtHCalc{\UniCProp}$ and $\UniFLeExtHCalc{\UniWProp}$ are obtained by the addition, respectively, of the rules $\{(\UniCRule)\}$ and $\{(\UniIRule),(\UniORule)\}$, listed below.
We follow the standard convention here of writing $(\UniWRule)$ to refer to the two rules $(\UniIRule)$ and $(\UniORule)$.
\begin{center}
\begin{small}
                \AxiomC{$H \VL \UniSequent{\UniMSetFmA,\UniMSetFmB}{\UniMSetSucA}$}
                \RightLabel{($\UniIRule$)}
                \UnaryInfC{$
                H \VL \UniSequent{\UniMSetFmA,\UniFmA,\UniMSetFmB}{\UniMSetSucA}$}
                \DisplayProof
                \; 
    \AxiomC{$H \VL \UniSequent{\UniMSetFmA}{}$}
                \RightLabel{($\UniORule$)}
                \UnaryInfC{$
                H\VL\UniSequent{\UniMSetFmA}{0}$}
                \DisplayProof  
                \;
                \AxiomC{$H \VL\UniSequent{\UniMSetFmA,
    \UniFmA,\UniFmA,\UniMSetFmB}{\UniMSetSucA}$}
                \RightLabel{($\UniCRule$)}
                \UnaryInfC{$H \VL \UniSequent{\UniMSetFmA,\UniFmA,\UniMSetFmB}{\UniMSetSucA}$}
                \DisplayProof
\end{small}
\end{center}

A calculus $\UniHyperCalcA_1$ is said to \emph{extend} (or be an \emph{extension of}) the calculus $\UniHyperCalcA_2$ if it contains all the rules of $\UniHyperCalcA_2$.
If $\UniHyperCalcA$
is an extension of
$\UniFLeExtHCalc{}$,
the \textit{extension of $\UniHyperCalcA$} by a set~$\UniSetRuleSchA$ of rule schemas is the hypersequent calculus~$\UniHyperCalcA \cup \UniSetRuleSchA$, which we write as~$\UniCalcExt{\UniHyperCalcA}{\UniSetRuleSchA}$ following standard practice.

Let $\UniHyperCalcA$ be an extension of $\UniFLeExtHCalc{}$. A \emph{derivation} in $\UniHyperCalcA$ is a finite labelled rooted tree such that the labels of a non-leaf node and its child node(s) are respectively the conclusion and premise(s) of some rule in the calculus.
A \textit{proof} of a hypersequent~$\UniHypersequentA$
in $\UniHyperCalcA$
is a derivation
in $\UniHyperCalcA$
whose root is labelled with $h$
and all leaves
are instances of axiomatic rules.
We write $\vdash_{\UniHyperCalcA} h$
to mean that there is a proof of $h$ in $\UniHyperCalcA$ (we say that $h$ is \emph{provable} in $\UniHyperCalcA$).
A \textit{branch} in a derivation is a path from the root to a leaf.
The \textit{height} of a derivation is the maximum number of nodes on a branch. We write $\vdash_{\UniHyperCalcA}^\alpha h$ (resp., $\vdash_{\UniHyperCalcA}^{\leq\alpha} h$)
to mean that there is a proof of $h$ in $\UniHyperCalcA$ of height $\alpha$ (resp., $\leq\alpha$).
We say that \emph{$\UniHyperCalcA$
is a calculus for a logic
$\vdash$} whenever
$\vdash \UniFmA$
iff $\vdash_{\UniHyperCalcA} (\UniSequent{}{\UniFmA})$.
In that case we also say that \emph{$\vdash$ is determined by $\UniHyperCalcA$},
and by \emph{provability in $\UniHyperCalcA$},
we mean the provability problem for $\vdash$.
We say that $\UniHyperCalcA$ is \emph{analytic} 
when for any hypersequent $h$ it proves, there is a proof of $h$ in $\UniHyperCalcA$ using only subformulas of $h$.

\begin{figure*}
\scriptsize
    \textbf{Initial hypersequents}
    \hspace{20em}
    \textbf{Structural rules}
    
    \vspace{.5em}
    {
    \begin{tabular}{ccc}
        \AxiomC{}
        \UnaryInfC{
        $\UniHyperMSetA 
        \VL 
        \UniSequent{\UniSchPropA}{\UniSchPropA}$
        }
        \DisplayProof
        &
        \AxiomC{}
        \UnaryInfC{
        $\UniHyperMSetA 
        \VL   
        \UniSequent{0}{}$
        }
        \DisplayProof
        &
        \AxiomC{}
        \UnaryInfC{
        $\UniHyperMSetA 
        \VL   
        \UniSequent{}{1}$
        }
        \DisplayProof
        \end{tabular}
    }\hspace{5em}
    {
            \begin{tabular}{rr}
                \AxiomC{$\UniHyperMSetA \VL   
                \UniSequent{\UniMSetFmA}{\UniFmA} \VL \UniSequent{\UniMSetFmA}{\UniFmA}$}
                \RightLabel{(EC)}
                \UnaryInfC{$\UniHyperMSetA \VL  \UniSequent{\UniMSetFmA}{\UniFmA}$}
                \DisplayProof
                \quad
                \AxiomC{$\UniHyperMSetA$}
                \RightLabel{(EW)}
                \UnaryInfC{$\UniHyperMSetA \VL   
                \UniSequent{\UniMSetFmA}{\UniFmA}$}
                \DisplayProof\\[1em]
            \end{tabular}
        }
        
        \textbf{Logical rules}
        \vspace{.5em}
        
        \begin{center}
            \begin{tabular}{cccccc}
            \AxiomC{$\UniHyperMSetA 
            \VL   
            \UniSequent{\UniMSetFmA,\UniFmA_i,\UniMSetFmB}{\UniMSetSucA}$}
            \RightLabel{($\land$L)}
            \UnaryInfC{$\UniHyperMSetA \VL   
            \UniSequent{\UniMSetFmA,\UniFmA_1\land \UniFmA_2,
            \UniMSetFmB}{\UniMSetSucA}$}
            \DisplayProof
            &
            \AxiomC{$\UniHyperMSetA \VL   
            \UniSequent{\UniFmA,\UniMSetFmA}{\UniFmB}$}
            \RightLabel{(${\imp}$R)}
            \UnaryInfC{$\UniHyperMSetA \VL   
            \UniSequent{\UniMSetFmA}{\UniFmA \imp \UniFmB}$}
            \DisplayProof
            &
            \AxiomC{$
            \UniHyperMSetA \VL \UniSequent{\UniMSetFmA}{}
            $}
            \RightLabel{($0$R)}
            \UnaryInfC{$\UniHyperMSetA 
            \VL
            \UniSequent{\UniMSetFmA}{0}$}
            \DisplayProof
            &
             \AxiomC{$
            \UniHyperMSetA 
            \VL   
            \UniSequent{\UniMSetFmA,\UniMSetFmB}{\UniMSetSucA}$}
            \RightLabel{($1$L)}
            \UnaryInfC{$\UniHyperMSetA 
            \VL   
            \UniSequent{\UniMSetFmA,1,\UniMSetFmB}{\UniMSetSucA}$}
            \DisplayProof 
            &
            \AxiomC{$\UniHyperMSetA \VL   
            \UniSequent{\UniMSetFmA}{\UniFmA}$}
            \AxiomC{$\UniHyperMSetA \VL   
            \UniSequent{\UniMSetFmB,\UniFmB,\UniMSetFmC}{\UniMSetSucA}$}
            \RightLabel{(${\imp}$L)}
            \BinaryInfC{$\UniHyperMSetA \VL   
            \UniSequent{\UniMSetFmB,\UniMSetFmA,\UniFmA\imp\UniFmB,\UniMSetFmC}{\UniMSetSucA}$}
            \DisplayProof
            &
            \AxiomC{$\UniHyperMSetA \VL 
            \UniSequent{\UniMSetFmA}{\UniFmA}$}
            \AxiomC{$\UniHyperMSetA \VL 
            \UniSequent{\UniMSetFmB}{\UniFmB}$}
            \RightLabel{($\fus$R)}
            \BinaryInfC{$\UniHyperMSetA \VL   
            \UniSequent{\UniMSetFmA,\UniMSetFmB}{\UniFmA\fus \UniFmB}$}
            \DisplayProof
            \\[0.4cm]
            \multicolumn{6}{c}{
            \begin{tabular}{ccccc}
            \AxiomC{$\UniHyperMSetA \VL  
            \UniSequent{\UniMSetFmA,\UniFmA,\UniMSetFmB}{\UniMSetSucA}$
            }
            \AxiomC{$\UniHyperMSetA \VL   
            \UniSequent{\UniMSetFmA,\UniFmB,\UniMSetFmB}{\UniMSetSucA}
            $
            }
            \RightLabel{($\lor$L)}
            \BinaryInfC{$\UniHyperMSetA \VL  
            \UniSequent{\UniMSetFmA,\UniFmA\lor\UniFmB,\UniMSetFmB}{\UniMSetSucA}$}
            \DisplayProof
            &
            \AxiomC{$\UniHyperMSetA \VL   
            \UniSequent{\UniMSetFmA}{\UniFmA_i}
            $}
            \RightLabel{($\lor$R)}
            \UnaryInfC{$\UniHyperMSetA 
            \VL   
            \UniSequent{\UniMSetFmA}{\UniFmA_1\lor\UniFmA_2}$}
            \DisplayProof
            &
            \AxiomC{$\UniHyperMSetA 
            \VL   
            \UniSequent{\UniMSetFmA,\UniFmA_i,\UniMSetFmB}{\UniMSetSucA}$}
            \RightLabel{($\land$L)}
            \UnaryInfC{$\UniHyperMSetA \VL   
            \UniSequent{\UniMSetFmA,\UniFmA_1\land \UniFmA_2,
            \UniMSetFmB}{\UniMSetSucA}$}
            \DisplayProof
            &
            \AxiomC{$\UniHyperMSetA \VL   
            \UniSequent{\UniMSetFmA}{\UniFmA}$}
            \AxiomC{$\UniHyperMSetA \VL   
            \UniSequent{\UniMSetFmA}{\UniFmB}$}
            \RightLabel{($\land$R)}
            \BinaryInfC{$\UniHyperMSetA \VL   
            \UniSequent{\UniMSetFmA}{\UniFmA\land \UniFmB}$}
            \DisplayProof
            &
             \AxiomC{$\UniHyperMSetA 
             \VL   
             \UniSequent{\UniMSetFmA,\UniFmA,\UniFmB,\UniMSetFmB}{\UniMSetSucA}$}
             \RightLabel{($\fus$L)}
             \UnaryInfC{$\UniHyperMSetA 
             \VL   
             \UniSequent{\UniMSetFmA,\UniFmA\fus \UniFmB,
             \UniMSetFmB}{ \UniMSetSucA}$}
             \DisplayProof
            \end{tabular}
            }
        \end{tabular}
    \end{center}
    \caption{The hypersequent calculus~$\UniFLeExtHCalc{}$
    for~$\UniFLeExtLogic{}$.}
    \label{figure-HFLec}
\end{figure*}

Ciabattoni et al.~\cite{CiaGalTer08} provided analytic hypersequent calculi for infinitely many axiomatic extensions of the logic~$\UniFLeExtLogic{}$ 
by adding \emph{analytic structural hypersequent rules} to $\UniFLeExtHCalc{}$.
In terms of the substructural hierarchy,
these logics correspond to acyclic axioms in level $\mathcal{P}_3'$ for extensions of $\UniFLeExtLogic{\UniCProp}$,
and in level $\mathcal{P}_3$
for extensions of $\UniFLeExtLogic{\UniWProp}$.
For an index set~$J \UniSymbDef \{j_{1},\ldots,j_{k}\}$, denote the hypersequent $\UniHyperMSetA \VL \UniSequent{\UniMSetFmA_{j_{1}}}{\UniMSetSucA_{j_{1}}} \VL \ldots \VL \UniSequent{\UniMSetFmA_{j_{k}}}{\UniMSetSucA_{j_{k}}}$ by $\UniHyperMSetA\VL \UniSequent{\UniMSetFmA_{j}}{ \UniMSetSucA_{j}} (j\in J)$.

\begin{definition}\label{def:ana-struct-rule}
An \emph{analytic structural rule schema} (henceforth, \emph{analytic rule}) has the following form,
where 
all symbols in antecedents are multiset-variables,
$I, J$ and $L$ are sets of indices,
$K_i$ is a set of indices for each $i \in I$, $\alpha_i \in \mathbb{N}$
for each $i \in I$,
$\beta_j \in \mathbb{N}$ for each $j \in J$,
and each $\vec{\UniMSetFmA}_{ik}$
and $\vec{\UniMSetFmB}_l$
is a list of multiset-variables:
\begin{center}
    \AxiomC{
    $\{ \UniHyperMSetA \VL \UniSequent{Y_i,\vec{\Gamma}_{ik}}{\UniMSetSucA_i}\}_{i \in I, k \in K_i}$
    }
    \AxiomC{$\{ 
    \UniHyperMSetA \VL 
    \UniSequent{\vec{\Delta}_l}{}
    \}_{l \in L}$}
    \BinaryInfC{
    $H 
    \VL 
    \UniSequent{Y_i, X_{i1},\ldots,X_{i\alpha_i}}{\UniMSetSucA_i} \, (i \in I)
    \VL
    \UniSequent{Z_{j1},\ldots,Z_{j\beta_j}}{} (j \in J)
    $}
    \DisplayProof
\end{center}
Moreover, it satisfies, among others~\cite{revantha2020}, the following properties:
(i) \emph{linear conclusion} (the multiset-variables that occur in the conclusion occur exactly once) and
(ii) \emph{strong subformula property} (each multiset-variable in the premise occurs in the conclusion).
\end{definition}

\begin{example}[Communication]\
\begin{center}
            \AxiomC{$\UniHyperMSetA
            \VL 
            \UniSequent{\UniMSetFmB_1,\UniMSetFmA_1}{\UniMSetSucA_1}
            $}
            \AxiomC{$\UniHyperMSetA
            \VL 
            \UniSequent{\UniMSetFmB_2,\UniMSetFmA_2}{\UniMSetSucA_2}
            $}
            \RightLabel{$(com)$}
            \BinaryInfC{$\UniHyperMSetA \VL \UniSequent{\UniMSetFmB_2,\UniMSetFmA_1}{\UniMSetSucA_1}
            \VL
            \UniSequent{\UniMSetFmB_1,\UniMSetFmA_2}{\UniMSetSucA_2}
            $}
            \DisplayProof
\end{center}
The fuzzy logic $\mathbf{MTL}$
is determined by the calculus $\UniFLeExtHCalc{\UniWProp}(\{com\})$.
For more examples of such schemas, see~\cite{CiaGalTer08}.
\end{example}

We will always use $\UniAnaRuleSet$
to denote a finite set of analytic structural rule schemas.
In this way, logics determined by
calculi of the form
$\UniHFLecR$ and $\UniHFLewR$
will be our main focus here.
Observe that in such calculi every
rule schema
except for $\UniEC$ and $\UniEW$ 
have the following form, for some $n,m\geq 1$,
where $T_1,\ldots,T_m,S_1,\ldots,S_n$ are components,
$H$ is a hypersequent-variable. 
\begin{center}
\begin{small}
\AxiomC{$\UniHyperMSetA \VL T_1$}
\AxiomC{$\cdots$}
\AxiomC{$\UniHyperMSetA \VL T_m$}
\TrinaryInfC{$\UniHyperMSetA \VL S_1\VL\cdots \VL S_n$}
\DisplayProof
\end{small}
\end{center}

We will design algorithms that accept hypersequents as input.
For the size of the input, 
we employ the following norm, which is related to the length $\UniSizeHyper{h}$ of representation on irredundant hypersequents by a polynomial factor.

\begin{definition}[(Hyper)sequent norms]\label{def:norms-hyper-seq}
For any sequent $s$, we define $\UniNorm{s}{}$
as the maximum multiplicity of a formula in the antecedent of $s$.
For any hypersequent $h$,
we define
$\UniNorm{h}{}\UniSymbDef\max_{s \in h} \UniNorm{s}{}$ (i.e., the max of the norms of its components). As we shall later see, for our purposes, it would be helpful to think of $h$ as a sequence of sequents, with no repetitions. In that context, $\UniNorm{h}{}$ tracks the largest size of any sequent that we have seen so far. Note that for hypersequents with no repetitions of sequents, this is a proper norm. 
\end{definition}

\subsubsection{Invertible calculi}
\label{sec:inv-calc-hp-admiss}
In order to develop our proof search procedures, we modify
the usual hypersequent calculi
to allow components in the conclusion of a rule
to persist upwards. 

Fix a
finite set $\UniAnaRuleSet$ of analytic structural rule schemas
and let $\mathbf{H} \in \{ \UniHFLecR,\UniHFLewR \}$.
Given a rule schema~$\UniRuleSchemaA$
that is not
in $\mathcal{S} \UniSymbDef \{ \UniEC, \UniEW\}$, the \emph{invertible form $\inv{\UniRuleSchemaA}$ of $\UniRuleSchemaA$} is obtained by including the principal components in every premise.
The \emph{principal}
and \emph{active}
components
of
$\inv{\UniRuleSchemaA}$
are the principal and active components
in $\UniRuleSchemaA$.

\begin{example}
    The invertible
    form of $(com)$
    is
\begin{center}
\begin{small}
            \AxiomC{$h_c
            \VL 
            \UniSequent{\UniMSetFmB_1,\UniMSetFmA_1}{\UniMSetSucA_1}
            $}
            \AxiomC{$h_c
            \VL 
            \UniSequent{\UniMSetFmB_2,\UniMSetFmA_2}{\UniMSetSucA_2}
            $}
            \RightLabel{$\inv{com}$}
            \BinaryInfC{$h_c
            $}
            \DisplayProof
\end{small}
\end{center}
where $h_c = (\UniHyperMSetA \VL \UniSequent{\UniMSetFmB_2,\UniMSetFmA_1}{\UniMSetSucA_1}
            \VL
            \UniSequent{\UniMSetFmB_1,\UniMSetFmA_2}{\UniMSetSucA_2})$.
The principal components
of
$\inv{com}$ are 
$\UniSequent{\UniMSetFmB_2,\UniMSetFmA_1}{\UniMSetSucA_1}
            $ and $\UniSequent{\UniMSetFmB_1,\UniMSetFmA_2}{\UniMSetSucA_2}$,
            and the active components are 
            $\UniSequent{\UniMSetFmB_1,\UniMSetFmA_1}{\UniMSetSucA_1}$
            and 
            $\UniSequent{\UniMSetFmB_2,\UniMSetFmA_2}{\UniMSetSucA_2}$.
\end{example}

\begin{definition}[Invertible calculi]
    The calculus $\UniHFLecRinv$ (resp., $\UniHFLewRinv$)
    is obtained by replacing
    in $\UniHFLecR$
    (resp., in $\UniHFLewR$)
    every
    rule schema $\UniRuleSchemaA$
    not in $\mathcal{S}$
    by $\inv{\UniRuleSchemaA}$.
\end{definition}

The following is straightforward to prove; use $\UniEC$ and $\UniEW$.


\begin{lemma}\label{lem:original-to-invert}
$\UniHFLecRinv$
and $\UniHFLewRinv$
prove the same hypersequents as
$\UniHFLecR$
and $\UniHFLewR$,
respectively.
\end{lemma}

It is also clear that $\UniHFLecRinv$
and $\UniHFLewRinv$
satisfy
the subformula property,
because $\UniHFLecR$
and $\UniHFLewR$ do.

\subsection{Well-quasi-orders}


A \emph{well-quasi-order} (or wqo) is a structure $\mathbf{A} \UniSymbDef \langle A,\le_\mathbf{A} \rangle$ where $A$ is a set and $\le_\mathbf{A} \subseteq A \times A$ is a reflexive and transitive relation s.t. for every infinite sequence of elements $a_0,a_1,\dots,$ of $A$ there exists $i < j$
s.t. $a_i \le_\mathbf{A} a_j$. A \emph{proper norm} on a set $A$ is a function $\UniNorm{\cdot}{\mathbf{A}}: A \to \mathbb{N}$ such that
for every $n \in \mathbb{N}$, the set $\{a : \UniNorm{a}{\mathbf{A}} \le n\}$ is finite. A \emph{normed wqo} (or nwqo) is a structure
$\mathbf{A} \UniSymbDef \langle A,\le_\mathbf{A}, \UniNorm{\cdot}{\mathbf{A}} \rangle$ such that $\langle A,\le_\mathbf{A} \rangle$ is a wqo and $\UniNorm{\cdot}{\mathbf{A}}$ is a proper norm.

It is easy to see that the set of natural numbers $\mathbb{N}$ with the usual ordering and the norm being the identity function is a nwqo. It turns out that this observation can be generalized as follows. For any $d \ge 1$, let $\mathbb{N}^d$
denote the set of all $d$-tuples of natural numbers. Given $\AntVecA = (x_1,\dots,x_d)$ and $\AntVecB = (y_1,\dots,y_d)$ in $\mathbb{N}^d$ we say that $\mathbf{x} \le_{\mathbb{N}^d} \mathbf{y}$ if $x_i \le y_i$ for each $i$. 
Finally, we set the norm of an element $\mathbf{x} = (x_1,\dots,x_d)$ as $\UniNorm{x}{\mathbb{N}^d} = \max_{1 \le i \le d} x_i$.
By Dickson's lemma~\cite{dickson1913}, it is known that 
$\UniWqoModNatural^d \UniSymbDef \langle \mathbb{N}^d,\le_{\mathbb{N}^d},
\UniNorm{\cdot}{\mathbb{N}^d} \rangle$ is an nwqo 
for every $d$. This nwqo is called the product nwqo over $\mathbb{N}^d$.
In the sequel, when the dimension $d$ is clear from context, we will drop the subscript $\mathbb{N}^d$
from $\le_{\mathbb{N}^d}$ and $\UniNorm{\cdot}{\mathbb{N}^d}$. 
We will also make use of \emph{disjoint sums} of nwqos, in particular of nwqos
of the form $k \cdot \UniWqoModNatural^d$ ($k,d \in \mathbb{N}$), where
$(i, \AntVecA) \leq_{k \cdot \mathbb{N}^d} (j, \AntVecB)$
iff $i = j$ (i.e., $\AntVecA$ and $\AntVecB$ are in the same copy of 
$\mathbb{N}^d$)
and $\AntVecA \leq_{\mathbb{N}^d} \AntVecB$;
here,
the norm is defined by
$\UniNorm{(i,\AntVecA)}{k \cdot \mathbb{N}^d} \UniSymbDef \UniNorm{\AntVecA}{\mathbb{N}^d}$.

\subsubsection*{Controlled bad sequences} A \emph{control function} is any function $\UniControlFunctionA: \mathbb{N} \to \mathbb{N}$ that is monotone and satisfies $\UniControlFunctionA(x) \ge x$ 
for all $x \in \mathbb{N}$. Now, let $\UniControlFunctionA$ be any control function, let $n \in \mathbb{N}$ and let $\mathbf{A} \UniSymbDef \langle A,\le_\mathbf{A}, \UniNorm{\cdot}{\mathbf{A}} \rangle$ be any nwqo.
A sequence of elements $a_0,a_1,\dots$ over $A$ is called a \emph{$(\UniControlFunctionA,n)$-controlled bad sequence} iff
\begin{itemize}
    \item There is no $i < j$ such that $a_i \le_\mathbf{A} a_j$.
    \item For every $i$, $\UniNorm{a_i}{A} \le \UniControlFunctionA^i(n)$ where $\UniControlFunctionA^i$ denotes the $i$-fold composition of $\UniControlFunctionA$ with itself.
\end{itemize}

By definition of a wqo and by definition of the first property of a $(\UniControlFunctionA,n)$-controlled bad sequence, it follows that no such sequence can be infinite. It turns out that, by an application of K\"onig's lemma, for every $\UniControlFunctionA$, every $n$ and every nwqo $\mathbf{A}$, there
is a maximum length for $(\UniControlFunctionA,n)$-controlled bad sequences of $\mathbf{A}$. Hence, we can define 
a \emph{length function} $L_{\mathbf{A},\UniControlFunctionA}$ for $\mathbf{A}$ and $\UniControlFunctionA$ which maps a number $n$ to the maximum length of $(\UniControlFunctionA,n)$-controlled bad sequences of $\mathbf{A}$.

\subsubsection*{Length function bounds for $k \cdot \mathbb{N}^d$} 
The complexity upper bounds that we prove in this paper
will ultimately depend on length functions for nwoqs $k \cdot \UniWqoModNatural^d$ (for some $d,k \in \mathbb{N}$). Hence, 
we recall existing results on upper bounds for them. To state these upper bounds,
we first need to introduce the so-called \emph{fast-growing complexity classes}, which we shall do now.
We will only state those notions and facts that are needed here.

First, we define a hierarchy of functions $\{F_i\}_{i\leq\omega}$ as follows. Let $F_1(n) = 2n$ 
and let $F_{i+1}(n) = F_i^n(n)$. We then set $F_\omega(n) = F_n(n)$, where we can think of $\omega$ as an element that is larger than any natural number.
Using these functions, we can define a family of functions $\mathscr{F}_i$ for each $i \in \mathbb{N} \cup \{\omega\}$
by closing the $F_i$ function (and a bunch of other simple functions like the successor function) under composition
and limited primitive recursion. We let $\UniFGHOneAppLevel{i}$ denote the class of  functions given by 
$\bigcup_{j < i} \bigcup_{p \in \mathscr{F}_j} \{F_i(p(n))\}$. All that we require for our purposes regarding $\UniFGHOneAppLevel{i}$
are the following facts.

\begin{theorem}[{\cite{schmitz2016hierar}}]\label{thm:composition-fomega}
    If $f$ is a function in $\UniFGHOneAppLevel{\omega}$ and $g_1, g_2$ are primitive recursive functions, then the function $g_1 \circ f \circ g_2$ is also in $\UniFGHOneAppLevel{\omega}$.
\end{theorem}

\begin{theorem}[Length theorem for Dickson's lemma~\cite{figueira2011}]\label{the:length-theorem-dickson}
    For any control function $\UniControlFunctionA$,
    there is
    $U : \mathbb{N}^3 \to \mathbb{N}$
    in $\UniFGHOneAppLevel{\omega}$
    s.t. for any $k,d \geq 0$, $U(k,d,x)$ upper bounds the length function $L_{k \cdot \mathbb{N}^d,\UniControlFunctionA}(x)$.
\end{theorem}

Let $\UniFGHProbOneAppLevel{i}$ be the complexity class, defined as the set of decision problems that can be solved by a deterministic Turing machine in time
$f \in \UniFGHOneAppLevel{i}$. The distinction between determinism and non-determinism and time and space becomes irrelevant
for $\UniFGHProbOneAppLevel{i}$ with $i > 2$, because such classes are closed under exponential-time reductions. Of primary interest
to us is the class $\UniFGHProbOneAppLevel{\omega}$, 
whose members are said to have \emph{Ackermannian} complexity.
The hierarchy is, actually, more commonly defined for ordinals up 
to $\epsilon_0$~\cite{schmitz2016hierar}.
In particular, $\UniFGHProbOneAppLevel{\omega^\omega}$ is the class of \emph{hyper-Ackermannian} problems.


\section{Ackermannian upper bounds for analytic hypersequent extensions of $\UniFLeExtLogic{\UniCProp}$}
\label{sec:ub-flec}

\subsection{Refined calculi and their key properties}
\label{sec:refined-calculi-contraction}

Ramanayake~\cite{revantha2020} constructed
a refined version of
$\UniHFLecR$
in which $\UniEW$, $\UniEC$
and $(\UniCProp)$ 
are \emph{height-preserving admissible} (or \emph{hp-admissible}, for short), meaning that
whenever a premise of one of these rules has a proof with height $\alpha \in \mathbb{N}$,
 its conclusion has a proof of height at most $\alpha$. 
We begin by recalling such construction, and a key consequence of it: \emph{Curry's lemma}.

In order to facilitate the formulation of the refined calculus and the proof that it satisfies the desired hp-admissibility results, we define the binary relations
$\UniIntCtrRel{k}{\UniCProp}$,
$\UniExtCtrRel{k}{\text{(EC)}}$
and
$\UniIntExtCtrRel{k}{l}$ on hypersequents (for each $k, l \geq 0$)
in the following way.

\begin{itemize}
\item $\UniHypersequentA_1 \UniIntCtrRel{k}{\UniCProp} \UniHypersequentA_{2}$ iff $\UniHypersequentA_{2}$ can be obtained from~$\UniHypersequentA_1$ by applying some number of instances of $\UniCProp$ such that every formula in a component in~$\UniHypersequentA_{2}$ occurs   up to~$k$ times less in that component than in the corresponding component in~$\UniHypersequentA_1$.

\item $\UniHypersequentA_1 \UniExtCtrRel{l}{\text{(EC)}} \UniHypersequentA_{2}$ iff $\UniHypersequentA_{2}$ can be obtained from~$\UniHypersequentA_1$ by applying some number of instances of~$\text{(EC)}$ such that every component of $\UniHypersequentA_{2}$ occurs up to~$l$ times less in~$\UniHypersequentA_{2}$ than in~$\UniHypersequentA_1$.

\item $\UniHypersequentA_1\UniIntExtCtrRel{k}{l}\UniHypersequentA_{2}$ iff there exists $\UniHypersequentA'$ s.t.~$\UniHypersequentA_1 \UniIntCtrRel{k}{\UniCProp} \UniHypersequentA'$ and $\UniHypersequentA' \UniExtCtrRel{l}{\text{(EC)}} \UniHypersequentA_{2}$.
\end{itemize}

The following is a fact that comes easily from these definitions.

\begin{propositionrep}
\label{prop:descendands-in-c}
    If
    $\UniHypersequentA_1\UniIntExtCtrRel{k}{l}
    \UniHypersequentA_2$,
    then, for each 
    $\UniSequentA_1 \in \UniHypersequentA_1$,
    there is
    $\UniSequentA_2 \in \UniHypersequentA_2$
    such that
    $\UniSequentA_1 \UniIntCtrRel{k}{\UniCProp} \UniSequentA_{2}$.
    We call $\UniSequentA_2$ a \emph{descendant} of $\UniSequentA_1$
    in $\UniHypersequentA_2$.
\end{propositionrep}
\begin{proof}
By definition,
$\UniHypersequentA_1\UniIntExtCtrRel{k}{l}
    \UniHypersequentA_2$
means that there exists
$\UniHypersequentA'$ such that~$\UniHypersequentA_1 \UniIntCtrRel{k}{\UniCProp} \UniHypersequentA'$ and $\UniHypersequentA' \UniExtCtrRel{l}{\text{(EC)}} \UniHypersequentA_{2}$.
We fix a $\UniSequentA_1 \in \UniHypersequentA_1$ and note that it is enough to show that there is a descendant of $\UniSequentA_1$ in $\UniHypersequentA'$, since $\UniEC$ does not make any component disappear in derivations (it only reduces multiplicities).
We pick a derivation witnessing
$\UniHypersequentA_1 \UniIntCtrRel{k}{\UniCProp} \UniHypersequentA'$ and that contains only applications of $\UniCProp$.
Looking at it top-down, a component at node $i$ has an \emph{immediate descendant} at node $i+1$, which is either the same component (in case of being in the instantiation of the hypersequent variable) or a component coming from it via an application $\UniCProp$ (that is, the component is active in the rule application).
In any case, a component and its immediate descendant are related under $\UniIntCtrRel{k}{\UniCProp}$.
Starting from $s_1$, then, we can get a trace of descendants ending in a component
$\UniSequentA' \in \UniHypersequentA_1$ and the transitivity of $\UniIntCtrRel{k}{\UniCProp}$ guarantees that $\UniSequentA_1 \UniIntCtrRel{k}{\UniCProp} \UniSequentA'$.
\end{proof}

The idea 
of the refined calculus
is that we introduce to each rule instance a fixed amount of contractions (internal and external). This amount is determined by the syntactic form of the calculus, using the values $\UniFmlaMultFixed$
and $\UniActiveCompFixed$
defined below.

\begin{definition}[{\cite{revantha2020}}] 
\label{def:acn-fm}
The \emph{formula multiplicity} of a rule schema is the maximum of the number of elements
over all of the antecedents of the principal components.
The \emph{active component number} of a rule schema is the number of components in its conclusion (including the hypersequent variable).

The formula multiplicity $\UniFmlaMultFixed$ of a calculus $\UniHFLecR$ is the maximum
of the formula multiplicities of its rule schemas. Similarly, the active component number 
$\UniActiveCompFixed$ is the maximum of the active component numbers of its rule schemas.
\end{definition}

\begin{definition}[Refined contraction-calculus]\label{def:hp-admissible-calculi}
    Let $\UniHyperCalcA$ be a hypersequent calculus
    containing $\UniEC$
    and $(\UniCProp)$.
    Define
    $\UniHCalcAbsorb{\UniHyperCalcA}$
    as the set of
    rules
    $\UniHRuleAbsorb{\UniRuleSchemaA}$
    of the form
    $(h_1,\ldots,h_n)/g$,
    where
    $(h_1,\ldots,h_n)/h_0$
    is a rule instance of
    $\UniRuleSchemaA \in \UniHyperCalcA$
    and
    $\UniHypersequentA_{0} \UniIntExtCtrRel{\UniFmlaMultFixed-1}{\UniActiveCompFixed} \UniHypersequentB$.
    The instance with
    conclusion $h_0$
    is called
    the \emph{base instance}
    of $\UniHRuleAbsorb{\UniRuleSchemaA}$.
\end{definition}



Clearly, the conclusion~$\UniHypersequentB$ of a non-base instance is obtained from the conclusion~$\UniHypersequentA_{0}$ of the base instance by performing a restricted number of applications of $(\UniCProp)$ (via $\UniIntCtrRel{\UniFmlaMultFixed-1}{\UniCProp}$) and then a restricted number of applications of (EC) (via $\UniExtCtrRel{\UniActiveCompFixed}{\text{(EC)}}$).

\begin{lemma}[\cite{revantha2020}]
\label{lem-hp-admissible}
Let $\UniAnaRuleSet$ be a finite set of analytic structural rules.
\begin{enumerate}
    \item 
The calculi
$\UniHFLecR$ and~$\UniHFLecRAbsorb$ prove the same hypersequents.
    \item $\UniEW$, $\UniEC$
and
$(\UniCProp)$
are hp-admissible in~$\UniHFLecRAbsorb$.
\end{enumerate}
\end{lemma}

Now we obtain the so-called \emph{Curry's lemma} for hypersequents.

\begin{definition}
    Given hypersequents
    $g$ and $h$,
    define
    $g \UniHyperseqStrWqo h$
    iff
    $g$ is obtained from
    $h$ by applications of
    $\{\UniEW, \UniEC, (\UniCProp)\}$.
\end{definition}

\begin{lemma}[hypersequent Curry's lemma]\label{lem:curry-hyper-ext}
    Let $h'$ and $h$ be
    hypersequents
    s.t.
    $h' \UniHyperseqStrWqo h$.
    If $\vdash^\alpha_{\UniHFLecRAbsorb} h$,
    then
    $\vdash^{\leq\alpha}_{\UniHFLecRAbsorb} h'$.
\end{lemma}

From~\cite{BalLanRam21LICS},
we know that $\UniHyperseqStrWqo$ is
a wqo when restricted to
hypersequents over the same
support.
Then the whole point of Curry's lemma is to prove existence of \emph{$\UniHyperseqStrWqo$-minimal proofs} 
in $\UniHFLecRAbsorb$, that is, every provable hypersequent is witnessed by a proof in which every branch is a bad sequence with respect to 
$\UniHyperseqStrWqo$.
 Then a proof search procedure for 
$\UniHFLecRAbsorb$ looks
for minimal proofs with respect to this wqo.
Our new algorithm will not follow this path, because
$\UniHyperseqStrWqo$ is too complex and leads to $\UniFGHProbOneAppLevel{\omega^\omega}$ upper bounds, instead of $\UniFGHProbOneAppLevel{\omega}$. In order to design this algorithm, we need to combine 
the two transformations on hypersequent calculi described above, namely $\UniHFLecRinv$ and $\UniHFLecRAbsorb$
into one.

\subsection{Invertible refined calculi for contraction}
Our focus from now on
will be on calculi of the form
$\UniHFLecRAbsorbinv$,
meaning
the result of
applying the $\UniHCalcAbsorb{(\cdot)}$
transformation 
to $\UniHFLecRinv$.
The following result shows that provability is height-preserved via such transformation
with respect to 
$\UniHFLecRAbsorb$.

\begin{lemmarep}\label{lem-inv-calc-same-seq}
Let
$h$ be a hypersequent and $\alpha \in \mathbb{N}$. If
$\vdash_{\UniHFLecRAbsorb}^{\alpha} h$,
then
$\vdash_{\UniHFLecRAbsorbinv}^{\leq\alpha} h$,
and vice versa.
\end{lemmarep}
\begin{proof}
Both direction are proved by induction on the proof height. The base cases are immediate in both because axiomatic rules are in practice not affected by the inversion of the calculus.

\emph{(Left-to-right)}
Given a proof $\pi$ of $h$ in $\UniHFLecRAbsorb$ of height $\alpha$, let
$\UniHRuleAbsorb{\UniRuleSchemaA}$
be 
the last rule applied in $\pi$, say with premises $h_1, \ldots, h_n$, whose form is
 $h_i = g \VL s_i$.
Thus $h_i$ has a proof in $\UniHFLecRAbsorb$ of height $\alpha_i < \alpha$ for each $1 \leq i \leq n$. Take $h_0$ to be the conclusion of the base instance of $\UniHRuleAbsorb{\UniRuleSchemaA}$, which has the form $h_0 = g \VL C$, where $C$ is the multiset of principal components of $\UniRuleSchemaA$
(see the picture below, on the left).
Define $h'_i := h_i \VL C
= g \VL s_i \VL C$, and note that $h_i' \UniHyperseqStrWqo h_i$ by $\UniEW$.
By Lem~\ref{lem:curry-hyper-ext},
$h'_i$ has a proof of height
$\leq \alpha_i < \alpha$ for each $1 \leq i \leq n$.
By the IH, each such 
$h_i'$ has a proof in 
$\UniHFLecRAbsorbinv$ of height $\leq \alpha_i < \alpha$.
Now, applying the appropriate instance of $\UniHRuleAbsorb{\inv{\UniRuleSchemaA}}$, as shown on the right of the picture below, we obtain a proof of $h$ in $\UniHFLecRAbsorbinv$
with height $\leq \alpha$. Here $g \VL C\UniIntExtCtrRel{\UniFmlaMultFixed-1}{\UniActiveCompFixed}h$.

\begin{center}
\AxiomC{$g \VL s_1 \;\; \cdots \;\; g \VL s_n$}
\RightLabel{$\UniHRuleAbsorb{\UniRuleSchemaA}$}
\UnaryInfC{$g \VL C$}
\DisplayProof
\qquad
\AxiomC{$g \VL C \VL s_1  \;\; \cdots \;\; g \VL C \VL s_n$}
\RightLabel{$\UniHRuleAbsorb{\inv{\UniRuleSchemaA}}$}
\UnaryInfC{$g \VL C$}
\DisplayProof
\end{center}

\emph{(Right-to-left)}
Given a proof 
$\pi$
of $h$ in $\UniHFLecRAbsorbinv$ of height $\alpha$, let 
$\UniHRuleAbsorb{\inv{\UniRuleSchemaA}}$ be the last rule instance applied in $\pi$, say with premises $h_1=g \VL C \VL s_1 ,\ldots, h_n = g \VL C \VL s_n$, where $C$ is the multiset of principal components of the conclusion $h_0$ of the corresponding base instance,
i.e., $h_0 = g \VL C$, and recall that
$h_0 
\UniIntExtCtrRel{\UniFmlaMultFixed-1}{\UniActiveCompFixed}
h
$.
By the IH,
there are proofs in
$\UniHFLecRAbsorb$
of $h_i$ with height $\leq \alpha_i < \alpha$
for each $1 \leq i \leq n$.
We know that there is
a base rule instance 
of $\UniHRuleAbsorb{\UniRuleSchemaA}$
in $\UniHFLecRAbsorb$
with same premises and conclusion
$h'_0 = g \VL C \VL C$
(the hypersequent variable being instantiated with $g \VL C$ now).
This proof of $h'_0$ has height $\leq \alpha$.
But then
$h \UniHyperseqStrWqo 
h_0
\UniHyperseqStrWqo
h'_0$,
and then $h$ has a proof of
height $\leq \alpha$
in
$\UniHFLecRAbsorb$
by Lem~\ref{lem:curry-hyper-ext}.
If the last rule applied
was $\UniEC$ or $\UniEW$,
the result directly follows by IH and 
Lem~\ref{lem:curry-hyper-ext}.
\end{proof}

Then, by Lem~\ref{lem:curry-hyper-ext}
and
Lem~\ref{lem-inv-calc-same-seq}:
\begin{lemmarep}\label{lem:curry-for-invertible}
Let $h'$ and $h$ be
    hypersequents
    such that
    $h' \UniHyperseqStrWqo h$.
If
$\vdash_{\UniHFLecRAbsorbinv}^{\alpha} h$,
then
$\vdash_{\UniHFLecRAbsorbinv}^{\leq\alpha} h'$.
\end{lemmarep}
\begin{proof}
Assume that
$h' \UniHyperseqStrWqo h$
and 
$\vdash_{\UniHFLecRAbsorbinv}^{\alpha} h$.
By Lem~\ref{lem-inv-calc-same-seq},
we have
$\vdash_{\UniHFLecRAbsorb}^{\beta} h$ with $\beta \leq \alpha$,
and by Lem~\ref{lem:curry-hyper-ext}
we have
$\vdash_{\UniHFLecRAbsorb}^{\gamma} h'$ with $\gamma \leq \beta$,
and again by Lem~\ref{lem-inv-calc-same-seq},
we obtain
$\vdash_{\UniHFLecRAbsorbinv}^{\leq\gamma} h'$.
Since $\gamma \leq \alpha$, we are done.
\qedhere

\end{proof}

\subsection{Proof search procedure}

Throughout this section, 
we denote by
$\UniHyperCalcA$ 
a calculus of the form
$\UniHFLecRAbsorbinv$ for some finite set $\UniAnaRuleSet$ of analytic structural rules.
The core idea is to focus only on hypersequents that are themselves bad sequences in the sequent wqo. 

A \emph{linearisation} of a hypersequent is any sequence of its components
(preserving the multiplicities).
For example,
if $h = s_1 \VL s_2 \VL s_3$ is a hypersequent,
then a linearization of $h$ is any
sequence $(s_i,s_j,s_k)$, 
where $\{i,j,k\} = \{1,2,3\}$.

Next, we define a nwqo on $S$-sequents, where $S$ is a finite set of formulas. Fix an enumeration $S= \{\varphi_1,\dots,\varphi_d\}$. A multiset $\Gamma$ of such formulas can be viewed as a vector in $v(\Gamma) \in \mathbb{N}^d$, where the $i^{th}$ component is the number of occurrences of $\varphi_i$ in $\Gamma$. We say that $\Gamma_1 \le \Gamma_2$ iff 
$v(\Gamma_1) \le_{\mathbb{N}^d} v(\Gamma_2)$ according to the product wqo over $\mathbb{N}^d$.
Hence, the defined order among multisets of formulas is also a wqo.

Now, for two $S$-sequents $s_1, s_2$ we say that $s_1 \le_S s_2$ iff $s_1 := \Gamma_1 \Rightarrow \Pi_1$, $s_2 = \Gamma_2 \Rightarrow \Pi_2$, $\Pi_1 = \Pi_2$ and $\Gamma_1 \le \Gamma_2$. It is easily seen that 
this order on the $S$-sequents is also a wqo. Hence, if we now consider the norm for sequents defined in Def~\ref{def:norms-hyper-seq}, we get a nwqo for $S$-sequents, which we denote by $\mathbf{W}(S)$. 

We ultimately want to consider only those hypersequents in our proof search whose linearisation yields a controlled bad sequence of sequents (for some suitable control function). To that end, we say that a function $f$ is an \emph{$\UniHyperCalcA$-control} if $\UniNorm{h'}{} \leq f(\UniNorm{h}{})$ for every rule instance in $\UniHyperCalcA$ having~$h$ as conclusion and $h'$ as a premise (Here we are using the norm for hypersequents in Def~\ref{def:norms-hyper-seq}). It is easy to see that a primitive recursive $\UniHyperCalcA$-control exists---indeed, even a linear function---for every hypersequent calculus~$\UniHyperCalcA$ considered here.

\begin{lemmarep}\label{lem:control-hyper-calc}
The function
$f(x) = (\UniFmlaMultFixed^2\UniActiveCompFixed^2 \cdot x) + 1$ is an $\UniHyperCalcA$-control.
\end{lemmarep}
\begin{proof}
Let $\UniHRuleAbsorb{\inv{\UniRuleSchemaA}}$ be a rule of $\UniHyperCalcA$
with conclusion $h$. Let $h'$ be a premise of it. Hence, there must exist $g$ such that $h'=g|s$ and  $g \UniIntExtCtrRel{\UniFmlaMultFixed-1}{\UniActiveCompFixed} h$. By definition of 
$\UniIntExtCtrRel{\UniFmlaMultFixed-1}{\UniActiveCompFixed}$, it follows that
$\UniNorm{g}{} \leq (\UniFmlaMultFixed-1)(\UniActiveCompFixed) \UniNorm{h}{}$.
Furthermore, by definition of the rules from $\UniHyperCalcA$, we have that every multiset variable that was instantiated in  $\UniHRuleAbsorb{\inv{\UniRuleSchemaA}}$ for the sequent $s$ must also occur in the principal components of the conclusion. This implies then that $\UniNorm{s}{} \le (\UniActiveCompFixed-1)\UniNorm{g}{} + 1$, (where the addition of a 1 comes from the fact that logical rules can increase the multiplicity of a formula by 1). Combining this with 
$\UniNorm{g}{} \leq (\UniFmlaMultFixed-1)(\UniActiveCompFixed) \UniNorm{h}{}$ proves
that $\UniNorm{h'}{} \leq \UniActiveCompFixed^2 \UniFmlaMultFixed^2 \UniNorm{h}{} + 1$.
\end{proof}

Now, let~$\pi$ be a derivation of some hypersequent~$h_0$ in~$\UniHyperCalcA$, and let~$h_0,\ldots,h_n$ be the sequence of hypersequents along any branch~$\mathsf{B}$ in~$\pi$. Let~$r_k$ ($0<k\leq n$) be the corresponding rule instance in~$\pi$ with conclusion~$h_{k-1}$ and premise~$h_{k}$, and let $s_k$ be the active component in that premise.
For any linearisation~$\lin{h}_0$ of $h_0$, the sequence $\lin{h}_0,s_1,\ldots,s_k$ is called the \emph{$\lin{h}_0$-sequence of~$\mathsf{B}$}. 

Let $\lin{h}$ be any linearisation of a hypersequent~$h$
and let $S$ be the set of subformulas of $h$; let $\lin{h}(0)$ denote the first element in the sequence. Given a control function $f$,
a \emph{$f$-minimal proof of $\lin{h}$} in $\UniHyperCalcA$ is a proof of $h$ such that
the $\lin{h}$-sequence of every branch is a $(f,\UniNorm{\lin{h}(0)}{})$-controlled bad sequence
over $\mathbf{W}(S)$.
This is well-defined as proofs in 
$\UniHyperCalcA$
are cut-free and thus all  elements of the $\lin{h}$-sequences are $S$-sequents.


\begin{lemma}[Existence of minimal proofs]\label{lem-minimal-proofs-complete}
Let $f$ be any $\UniHyperCalcA$-control.
Let $\lin{h}$ be a linearization of a hypersequent $h$ that is a $(f,\UniNorm{\lin{h}(0)}{})$-controlled bad sequence over $\mathbf{W}(S)$, where $S$ is the set of subformulas of $h$.
Then, $h$ is provable in $\UniHyperCalcA$ iff $\lin{h}$ has an $f$-minimal proof in $\UniHyperCalcA$.
\end{lemma}
\begin{proof}
The right-to-left direction is trivial since an $f$-minimal proof of $\lin{h}$ is a proof of $h$.
For the left-to-right direction, suppose that $\pi$ is a proof of $h$
in $\UniHyperCalcA$. Induction on the height $\alpha$ of~$\pi$.

(\textit{Base case})
If $\pi$ has height~$1$, then $h$ must be an initial hypersequent. The latter is a proof containing a singleton branch. Consequently, the $\lin{h}$-sequence of this branch is simply $\lin{h}$. The latter is a $(f,\UniNorm{\lin{h}(0)}{})$-bad sequence by hypothesis, so the claim holds.

(\emph{Inductive step})
Suppose $\pi$ has height $\alpha>1$. Consider the last rule $\UniHRuleAbsorb{\UniRuleSchemaA}$ applied in $\pi$. Let~$\pi_i$ be the subproof of $i$-th premise of~$\UniHRuleAbsorb{\UniRuleSchemaA}$, with height $\alpha_i < \alpha$. Note that the root of $\pi_i$ has the form $g \VL C \VL s_i$,  where $s_i$ is the active component of the rule and $C$
is the multiset of principal components of the base instance of $\UniHRuleAbsorb{\UniRuleSchemaA}$. Here $g \VL C \UniIntExtCtrRel{\UniFmlaMultFixed-1}{\UniActiveCompFixed}h$.
\begin{center}
\begin{small}
    \AxiomC{$\pi_1$}
    \noLine
    \UnaryInfC{$g \VL C \VL s_1$}
    \AxiomC{$\ldots$}
    \AxiomC{$\pi_n$}
    \noLine
    \UnaryInfC{$g \VL C \VL s_n$}
    \RightLabel{$\UniHRuleAbsorb{\UniRuleSchemaA}$}
    \TrinaryInfC{$g \VL C$}
    \DisplayProof
\end{small}    
\end{center}

We then have that
$h \VL s_i \UniHyperseqStrWqo
g \VL C \VL s_i$
for each $1 \leq i \leq n$, and thus, by Lem~\ref{lem:curry-for-invertible},
there are proofs
$\pi'_i$
of each $h \VL s_i$
with height $\leq \alpha_i < \alpha$.
We now reason by cases on
the sequences
$\lin{h},s_i$:

\begin{enumerate}[leftmargin=*]
    \item $\lin{h},s_i$ is a $(f,\UniNorm{\lin{h}(0)}{})$-controlled bad sequence for every $i$: Applying the IH to $\pi'_i$, we obtain a $f$-minimal proof~$\pi_i''$ of $\lin{h},s_i$ for each~$i$. By definition, the $(\lin{h},s_i)$-sequence of every branch is a $(f,\UniNorm{\lin{h}(0)}{})$-controlled bad sequence.
We now claim that the following is a proof in $\UniHyperCalcA$, with 
$h \VL C \UniIntExtCtrRel{\UniFmlaMultFixed-1}{\UniActiveCompFixed}h$.
\begin{center}
\begin{small}
    \AxiomC{$\pi''_1$}
    \noLine
    \UnaryInfC{$h \VL C \VL s_1$}
    \AxiomC{$\ldots$}
    \AxiomC{$\pi''_n$}
    \noLine
    \UnaryInfC{$h \VL C \VL s_n$}
    \RightLabel{$\UniHRuleAbsorb{\UniRuleSchemaA}$}
    \TrinaryInfC{$h \VL C$}
    \DisplayProof
\end{small}    
\end{center}

For that, it is enough to show that 
$
h \VL C \UniIntExtCtrRel{\UniFmlaMultFixed-1}{\UniActiveCompFixed}
h
$.
We know that 
$
g \VL C \UniIntExtCtrRel{\UniFmlaMultFixed-1}{\UniActiveCompFixed}
h
$. Therefore, by Prop~\ref{prop:descendands-in-c},
for each $s' \in C$,
there is $s'' \in h$
such that $s' \UniIntCtrRel{\UniFmlaMultFixed-1}{\UniCProp}
s''$.
Thus, $h \VL C \UniIntCtrRel{\UniFmlaMultFixed-1}{\UniCProp} h \VL C'$,
where every component in $C'$ appears in $h$.
But $|C| \leq \UniActiveCompFixed$
because $C$ is the multiset of principal components of the base instance, and so 
$h \VL C' \UniExtCtrRel{\UniActiveCompFixed}{\UniEC} h$. Hence, $h \VL C \UniIntExtCtrRel{\UniFmlaMultFixed-1}{\UniActiveCompFixed} h$, as desired.

It only remains to argue that the above proof, which we call $\pi_0$, is an $f$-minimal proof of $\lin{h}$.
The $\lin{h}$-sequence of any branch in $\pi_0$ has the form $\lin{h},s_i,\ldots$ for some $i$. By inspection, it is the $(\lin{h},s_i)$-sequence of some branch in $\pi_i''$. We already observed that the latter is a $(f,\UniNorm{\lin{h}(0)}{})$-controlled bad sequence. Thus, $\pi_0$ is an $f$-minimal proof of $\lin{h}$.
    \item  $\lin{h},s_i$ is not a $(f,\UniNorm{\lin{h}(0)}{})$-controlled bad sequence for some~$i$: Note that $\lin{h}$ is a $(f,\UniNorm{\lin{h}(0)}{})$-controlled bad 
    sequence by assumption. 
    We now argue that this sequence must be $(f,\UniNorm{\lin{h}(0)}{})$-controlled
    (so that it must fail to be bad).
    Because $\UniNorm{g \VL C \VL s_i}{} \leq f(\UniNorm{h}{})$
    by Lem~\ref{lem:control-hyper-calc},
    and
    by definition of
    $\UniNorm{\cdot}{}$ we have $\UniNorm{s_i}{} \leq \UniNorm{g \VL C \VL s_i}{}$,
    we obtain
    $\UniNorm{s_i}{}
    \leq f(\UniNorm{h}{})$.
     By definition of the norms for hypersequents we get that 
     $\UniNorm{h}{} = \UniNorm{\lin{h}(j)}{}$ for some $0\leq j < |h|$.
    By the control property, we have
    $\UniNorm{\lin{h}(j)}{} \leq 
    f^j(\UniNorm{\lin{h}(0)}{})
    \leq f^{|h|-1}(\UniNorm{\lin{h}(0)}{})$.
    Thus $\UniNorm{s_i}{} \leq f(f^{|h|-1}(\UniNorm{\lin{h}(0)}{})) = f^{|h|}(\UniNorm{\lin{h}(0)}{})$,
    and $\lin{h},s_i$
    is controlled.
    So
    the sequence must fail to be bad, hence there
    exists $t\in\lin{h}$ s.t. $t$ can be obtained from $s_i$ by (zero or more) contractions. In either case, it follows that
$h \UniHyperseqStrWqo h|s_i$. By Lem~\ref{lem:curry-for-invertible} we obtain a proof~$\pi_0$ of $h$ of height $<\alpha$. Applying the IH to $\pi_0$, we obtain an $f$-minimal proof of $\lin{h}$.\qedhere
\end{enumerate}
\end{proof}


\begin{definition}[Proof-search tree]\label{def:minimal-proof-search-flec}
The \emph{minimal proof search tree} in $\UniHyperCalcA$ for the sequent $s$ is
defined as the limit tree of the following.
\begin{description}[{labelindent=\parindent},leftmargin=6em,style=nextline]
\item[$T_0(s)$] A single node labelled with $s$.
\item[$T_{n+1}(s)$] For each leaf~$\ell$ (labelled by the hypersequent~$h$, say) in $T_n(s)$, and each rule instance that has $h$ as conclusion,
add the premises as children of~$\ell$, provided tree so obtained is $f$-minimal.
\end{description}
\end{definition}

\begin{lemma}\label{lem:termination-flec}
    Given a sequent $s$,
    there is a least
    $N \in \mathbb{N}$ such that $T_N(s)=T_{N+1}(s)$.
    Moreover, 
    there is
    $U \in \UniFGHOneAppLevel{\omega}$
    such that
    $N$
    is upper bounded by 
$U(\UniNorm{s}{})$
for any sequent $s$.
\end{lemma}
\begin{proof}
    Let $S$ be the set of subformulas of $s$.
    By construction, every branch in $T_n(s)$ corresponds to a $(f,\UniNorm{s}{})$-controlled bad sequence over $\mathbf{W}(S)$. Suppose there is no  $N \in \mathbb{N}$ such that $T_N(s)=T_{N+1}(s)$. Then, consider the infinite limit tree $\bigcup_{n \geq 0}  T_n(s)$. This tree must be finitely branching,
    since every rule schema of $\UniHyperCalcA$ can only be instantiated in finitely many ways
    for every hypersequent. By K\"onig's Lemma there would exist an infinite branch, which induces an infinite bad sequence
    over $\mathbf{W}(S)$, what is absurd.
Moreover, this $N$ is uniformly upper bounded by a function in $\UniFGHOneAppLevel{\omega}$
because each step in the construction
adds a new element to a $(f,\UniNorm{s}{})$-controlled bad sequence producing a longer $(f,\UniNorm{s}{})$-controlled bad sequence, and by Thm~\ref{the:length-theorem-dickson} we know that the length of such a sequence is upper bounded by a function in $\UniFGHOneAppLevel{\omega}$
on $\UniNorm{s}{}$.
\end{proof}

From the above, the following is immediate.

\begin{theorem}\label{the:complexity-flec-ref-calc}
Provability in $\UniHyperCalcA$ is in $\UniFGHProbOneAppLevel{\omega}$.
\end{theorem}
\begin{proof}
    Given an input sequent $s$, by Lem~\ref{lem:termination-flec} we have that
    $T_N(s)$ can be constructed in space primitive recursive on
    $U \in \UniFGHOneAppLevel{\omega}$,
    which is thus also a function in $\UniFGHOneAppLevel{\omega}$ by Thm~\ref{thm:composition-fomega}. Furthermore, we can search among its subtrees for a minimal proof of $s$,
    which is again a primitive-recursive procedure by Thm~\ref{thm:composition-fomega}.
    By construction, any $f$-minimal proof of $s$ will be a subtree of $T_N(s)$. Hence, by Lem~\ref{lem-minimal-proofs-complete}, existence of such a subtree is equivalent to provability of $s$. 
\end{proof}

Finally, from 
Lem~\ref{lem:original-to-invert},
Lem~\ref{lem-inv-calc-same-seq}
and
Thm~\ref{the:complexity-flec-ref-calc}, we obtain:

\begin{theorem}[Main theorem I]
Provability in any extension of
$\UniFLeExtLogic{\UniCProp}$ 
by finitely many acyclic $\mathcal{P}'_3$-axioms
(i.e, provability in any analytic structural extension of $\UniFLeExtHCalc{\UniCProp}$)
has Ackermannian complexity.
\end{theorem}

\section{Ackermannian upper bounds for hypersequent extensions of $\UniFLeExtLogic{\UniWProp}$: $\omega$-calculi}
\label{sec:ub-flew-calculus}
\begin{toappendix}
\begin{remark}[A high-level description of the algorithm]\label{rem-high-level}
\emph{
Taking inspiration from the Karp-Miller algorithm for computing coverability sets in vector addition systems~\cite{KarpM69}, we undertake the following algorithm. Here we provide a high-level informal description.}

\emph{
Given an input hypersequent, perform the usual backward proof search, exhaustively applying the rules, until an active component is encountered (in a premise) that is larger than some component that is already present. Instead of writing down this new active component, we write down a revised version, by replacing every $\psi^k$ in it that had strictly larger multiplicity than its counterpart (in the already-present component) with~$\psi^\omega$. So our proof search is built on $\omega$-hypersequents, itself built from $\omega$-components, where the antecedent consists of an $\omega$-set---the set of $\omega$-formulas---and a vector in $\mathbb{N}^d$ which holds the information about the formulas with finite multiplicities. A formula is an $\omega$-formula or has finite multiplicity (but not both).}

\emph{
We need to show three things: it is possible to obtain a usual proof from the $\omega$-proof (soundness), the proof search terminates, and it finds a proof whenever there is one to be found (completeness).}

\emph{
For soundness, the idea is that the path from an $\omega$-hypersequent where its active component strictly increased its $\omega$-set down to its $\omega$-partner below can be iterated as often as required, to replace, in the usual proof we are now building, $\psi^\omega$ with $\psi^k$ for suitably large $k$. We can initially work out the required value by a leaf-to-root analysis of the $\omega$-proof (since it is finite). It is more helpful to see this path from the rule instance introducing the $\omega$-partner to the revised active component with strictly larger $\omega$-set. To make this iteration work, we need to be able to pass the additional copies of formula~$\psi$ that we create, so the multiplicity grows after each iteration. For this to work, we need a stronger relationship between the $\omega$-partner and the revised active component: we need the two components to be related in a proof-theoretic sense; informally, the pair should be in the transitive closure of the relation $R$ of components that are instantiations of components in the rule schema that share a multiset-variable (it is this variable that will be used to pass the additional formulas along).}

\emph{
If the proof search does not terminate, that is because the proof search tree is infinite. We can easily extract an infinite $R$-tree from this. To get a contradiction, we need to extract an infinite branch on this tree, and for that we need to ensure that $R$ is finitely branching. The way we described $R$ in the previous paragraph does not ensure finite branching: after all, even the input hypersequent could be used infinitely often along a branch of proof search. The solution is to be even more stringent on the $R$ relation. Not only should it relate a principal component with an active component, it must relate the ``best'' principal component. The answer to what is best turns out to be the principal component that was created last and has the possibility to change the active component in the question (we call it a \emph{key ancestor}). Now, if this principal component were related to infinitely many components under the transitive closure, there must be a repeat since the best principal component controls how many different components can be created while remaining the best.}

\emph{
For completeness, it suffices to show that whatever is done in a usual proof can be simulated in the $\omega$-proof. Since $\psi^\omega$ has a natural meaning as infinitely many copies of $\psi$, the suitable rules for $\omega$-formulas are quite clear to deduce. In particular, $\omega$-formulas persist upwards.}\qed
\medskip

\end{remark}
\end{toappendix}

Taking inspiration from the Karp-Miller algorithm for computing coverability sets in vector addition systems~\cite{KarpM69}, in backward proof search,
if the active component of a premise is strictly larger than an earlier component, 
we record this increment by replacing every formula~$\UniFmA$ that has strictly larger multiplicity with $\UniFmA^\omega$. 
We refer to it as an \emph{$\omega$-formula}, which is a member of
the extended language
$
L_\omega \UniSymbDef 
L \cup \{ \UniFmA^\omega \mid \UniFmA \in L \}$. NB. In this language we allow e.g., $(\UniFmA \land \UniFmB)^\omega$, but not
$\UniFmA^\omega \land \UniFmB^\omega$.
Once a formula becomes an $\omega$-formula, the proof rules ensure that it remains so from that point on. 
\smallskip

\noindent\fbox{An high-level description appears as Remark~\ref{rem-high-level} in the Appendix.}
\smallskip



Throughout this and the remaining sections, as usual,
we use $\UniAnaRuleSet$ to mean
an arbitrary finite set of
analytic hypersequent structural rule schemas.
Since $\UniHFLewRinv$ satisfies the subformula property, a proof of an input hypersequent can be restricted to the finite set of its subformulas, say $\UniSubfmlaHyperseqSet \UniSymbDef \{ \UniFmA_1,\ldots,\UniFmA_d\}$.
For the latter, we use $\mathbb{N}_{\leq d} \UniSymbDef\{0,1,\ldots, d\}$, where $0$
represents the empty stoup.
I.e., a formula is represented by its index in a fixed enumeration of $\UniSubfmlaHyperseqSet$.
We also let $\mathbb{N}_{\leq d}^+ \UniSymbDef \{1,\ldots, d\}$.
Also, 
as in the previous section,
we read a multiset
over $S$ as a vector
$\AntVecA \in \mathbb{N}^d$,
such that $\AntVecA(a)$ is
the multiplicity of the formula (with index) $a$
in the multiset.
This motivates the following formalization.

\begin{definition}[$\omega$-(hyper)sequents]\label{def:omega-hyper}
Let $d \geq 1$.
An element $(\AntOmSetA;\AntVecA)\in \mathcal{P}(\mathbb{N}_{\leq d}^+)\times \mathbb{N}^d$ such that $\AntOmSetA\cap \{i \in \mathbb{N}_{\leq d}^+ \VL \AntVecA(i)>0\} = \varnothing$ is called an \emph{$(\omega,d)$-antecedent}.
Elements in $\AntOmSetA$ (resp., $\AntVecA$) are called 
\emph{$\omega$-coordinates}.
An \emph{$(\omega,d)$-sequent} is an expression $(\AntOmSetA;\AntVecA)\Ra b$, where $(\AntOmSetA;\AntVecA)$ is an $(\omega,d)$-antecedent and $b\in\mathbb{N}_{\leq d}$. We call $\AntOmSetA$
the \emph{$\omega$-set} of the $\omega$-sequent.
%
%
An \emph{$(\omega,d)$-hypersequent} is a multiset of $(\omega,d)$-sequents.
We may omit the $d$ from the notation when it can be inferred from context.
An $\omega$-sequent is \emph{$\omega$-free}
when its $\omega$-set is empty, and an $\omega$-hypersequent is \emph{$\omega$-free} when all of its components are $\omega$-free.
\end{definition}
Observe that an $(\omega,d)$-antecedent $(\AntOmSetA;\AntVecA)$ has the property that every formula $a\in \mathbb{N}_{\leq d}^+$ is either an $\omega$-formula (i.e., $a \in \AntOmSetA$ and $\AntVecA(i)=0$), or it is not an $\omega$-formula ($a\not\in \AntOmSetA$).

\begin{notation}\label{not:vector-coord-repl} 
Given $\kappa\subseteq\mathbb{N}_{\leq d}^+$, $k \in \mathbb{N}$
and $\AntVecA\in\mathbb{N}^d$,
let $\AntVecA[\kappa\mapsto k]$ denote the vector obtained by setting the coordinates of $\AntVecA$
in $\kappa$
to $k$ and leaving the other coordinates unchanged. 
For $\AntVecA\in\mathbb{N}^d$
and $1 \leq a \leq d$, we write $\AntVecA,a$ to mean $\AntVecA+\mathbf{e}_a$, where $\mathbf{e}_a$ is the $a$-th unit vector.
\end{notation}


\begin{definition}[Adding formulas to $(\AntOmSetA;\AntVecA)$]\label{def:adding-formulas-omega-ant}
Define the $\omega$-antecedent $((\AntOmSetA;\AntVecA),a)$ obtained by the addition of the formula~$a$ to $(\AntOmSetA;\AntVecA)$ to be $(\AntOmSetA;\AntVecA)$ if $a\in \AntOmSetA$, else it is $(\AntOmSetA;\AntVecA,a)$.
For a multiset $I=\UniMultiset{a_1,\ldots,a_k}$ of formulas, define $(\AntOmSetA;\AntVecA),I$ as $(\ldots((\AntOmSetA;\AntVecA),a_1),\ldots),a_k$.
This does not depend on the ordering in the enumeration of $I$ so it is well-defined. 
\end{definition}

In words, the addition of the formula $a$ to $(\AntOmSetA;\AntVecA)$ leaves the latter unchanged if $a$ is already an $\omega$-formula (i.e., $a \in \AntOmSetA$), else its finite multiplicity is incremented (i.e., $\AntVecA$ becomes $\AntVecA,a$).

\subsection{The calculus $\UniHFLewRinv_\omega$}


Our first calculus dealing with $\omega$-formulas
is denoted $\UniHFLewRinv_\omega$, and its rule instances are presented in Fig.~\ref{fig-omega-rules}.
In order to understand it fully, we need to present the notion of \emph{$\omega$-structural rules}.

Every $\UniRuleSchemaA \in \UniAnaRuleSet$
has the form in Def~\ref{def:ana-struct-rule}, so
every principal component
corresponds to an index in $I \cup J$,
and every active component corresponds either to an index $(i,k)$ with $i \in I$
and $k \in K_i$,
or to an index $l \in L$.
For $i \in I$
and $k \in K_i$,
let $C_{i,k}$ be the set of those principal component indices $v \in I \cup J$
that have a multiset metavariable in common with the active component $(i,k)$.
Define $C_{l}$
for each active component $l \in L$ analogously.
Finally, define
$\AntOmSetA[i,k]$
as $\bigcup_{v \in C_{i,k}} \AntOmSetA_v$
for each $(i,k)$,
and $\AntOmSetA[l]$
as $\bigcup_{v \in C_{l}} \AntOmSetA_v$
for each $l \in L$.

\begin{definition}[$\omega$-structural rules, $\UniAnaRuleSet_\omega$]\label{def:omega-struct-rules}
If $\UniRuleSchemaA$ is an 
analytic structural hypersequent rule schema,
then $\UniRuleSchemaA_\omega$ is the $\omega$-version below, where $h_c$ is the conclusion meta-hypersequent.
\begin{center}
    \begin{footnotesize}
    \AxiomC{
    $\{ h_c \VL \UniSequent{(\AntOmSetA[i,k]; \varnothing),Y_i,\vec{\Gamma}_{ik}}{\UniMSetSucA_i}\}_{i \in I, k \in K_i}$
    }
    \AxiomC{$\{ 
    h_c \VL 
    \UniSequent{(\AntOmSetA[l]; \varnothing),\vec{\Delta}_l}{}
    \}_{l \in L}$}
    \RightLabel{$\UniRuleSchemaA_\omega$}
    \BinaryInfC{
    $H 
    \VL 
    \UniSequent{(\AntOmSetA_i;Y_i, X_{i1},\ldots,X_{i\alpha_i})}{\UniMSetSucA_i} \, (i \in I)
    \VL
    \UniSequent{(\AntOmSetA_j;Z_{j1},\ldots,Z_{j\beta_j})}{} (j \in J)
    $}
    \DisplayProof
    \end{footnotesize}
\end{center}
\end{definition}

We remind the reader that an $\omega$-hypersequent is a multiset of $\omega$-sequents of the form $(\AntOmSetA;\AntVecA)\Ra b$.
The proof calculus $\UniHFLewRinv_\omega$ in Fig.~\ref{fig-omega-rules} operates on $(\omega,d)$-hypersequents. For example, in addition to the usual logical rules---in the present setting, these operate on those formulas with finite multiplicities (described by the vector~$\AntVecA$)---there are also logical rules that operate on the $\omega$-coordinates~$\AntOmSetA$. As a concrete example, consider the rule schema
\begin{center}
\begin{footnotesize}
\AxiomC{$h_c \VL \UniSequent{(\AntOmSetA\cup\{\UniFmA\imp\UniFmB\};\UniMSetFmA)}{\UniFmA}$}
  \AxiomC{$h_c \VL \UniSequent{(\AntOmSetA\cup\{\UniFmA\imp\UniFmB\};\UniMSetFmB),\UniFmB}{\UniMSetSucA}$}
  \RightLabel{($\imp$L)}
  \BinaryInfC{$\UniHyperMSetA \VL \UniSequent{(\AntOmSetA\cup\{\UniFmA\imp\UniFmB\};\UniMSetFmB,\UniMSetFmA)}{\UniMSetSucA}$}
  \DisplayProof
\end{footnotesize}
\end{center}
Observe that this rule schema operates on a component that has $\UniFmA\imp\UniFmB$ among its $\omega$-coordinates; we also refer to this by saying that $\UniFmA\imp\UniFmB$ is an $\omega$-formula in that component. In this rule schema, the $\omega$-coordinates of the active components (in the premise) are the same as the $\omega$-coordinates in the principal component (in the conclusion).
Consider the right premise of the above. Since we do not know whether $\UniFmB\in \AntOmSetA$, we use Def~\ref{def:adding-formulas-omega-ant} to ensure that the formula is classified correctly; this also ensures that we maintain the property that, in any component, the $\omega$-coordinates are disjoint from the formulas with positive finite multiplicity.
 In general, $\omega$-coordinates increase from principal to active components, e.g., as seen by the $(\AntOmSetA[i,k]; \varnothing),Y_i,\vec{\Gamma}_{ik}$ term in Def~\ref{def:omega-struct-rules}. 

\begin{figure*}
\centering
\footnotesize
\textbf{Initial hypersequents}\par\medskip
\begin{tabular}{@{}ccccc@{}}
  \AxiomC{}
  \UnaryInfC{$\UniHyperMSetA \VL \UniSequent{(\AntOmSetA;\UniSchPropA)}{\UniSchPropA}$}
  \DisplayProof
  &
  \AxiomC{}
  \UnaryInfC{$\UniHyperMSetA \VL \UniSequent{(\AntOmSetA\cup\{\UniSchPropA\};\varnothing)}{\UniSchPropA}$}
  \DisplayProof
  &
  \AxiomC{}
  \UnaryInfC{$\UniHyperMSetA \VL \UniSequent{(\AntOmSetA;0)}{}$}
  \DisplayProof
  &
  \AxiomC{}
  \UnaryInfC{$\UniHyperMSetA \VL \UniSequent{(\AntOmSetA\cup\{0\};\varnothing)}{}$}
  \DisplayProof
  &
  \AxiomC{}
  \UnaryInfC{$\UniHyperMSetA \VL \UniSequent{(\AntOmSetA;\varnothing)}{1}$}
  \DisplayProof
\end{tabular}

\par\bigskip
\textbf{Structural rules}\par\medskip
\begin{tabular}{@{}cc@{}}
  \AxiomC{$h_c \VL \UniSequent{(\AntOmSetA;\UniMSetFmA)}{\UniMSetSucA}$}
  \RightLabel{w}
  \UnaryInfC{$\UniHyperMSetA \VL \UniSequent{(\AntOmSetA;\UniMSetFmA,\UniFmA)}{\UniMSetSucA}$}
  \DisplayProof
  &
  \AxiomC{$\UniHyperMSetA \VL \UniSequent{(\AntOmSetA;\UniMSetFmA)}{\UniFmA} \VL \UniSequent{(\AntOmSetA;\UniMSetFmA)}{\UniMSetSucA}$}
  \RightLabel{(EC)}
  \UnaryInfC{$\UniHyperMSetA \VL \UniSequent{(\AntOmSetA;\UniMSetFmA)}{\UniMSetSucA}$}
  \DisplayProof
  \qquad
  \AxiomC{$\UniHyperMSetA$}
  \RightLabel{(EW)}
  \UnaryInfC{$\UniHyperMSetA \VL \UniSequent{(\AntOmSetA;\UniMSetFmA)}{\UniMSetSucA}$}
  \DisplayProof
\end{tabular}
\begin{center}
    All rules schemas in
    $\UniAnaRuleSet_\omega$
    (see Def.~\ref{def:omega-struct-rules}).
\end{center}

\par\bigskip
\textbf{Logical rules operating on $\omega$-coordinates}\par\medskip
\begin{tabular}{ccc}
  \AxiomC{$h_c \VL \UniSequent{(\AntOmSetA;\UniMSetFmA)}{\UniMSetSucA}$}
  \RightLabel{($1$L)}
  \UnaryInfC{$\UniHyperMSetA \VL \UniSequent{(\AntOmSetA \cup \{1\};\UniMSetFmA)}{\UniMSetSucA}$}
  \DisplayProof
  &
  \AxiomC{$h_c \VL \UniSequent{(\AntOmSetA\cup\{\UniFmA\fus \UniFmB\};\UniMSetFmA),\UniFmA,\UniFmB}{\UniMSetSucA}$}
  \RightLabel{($\fus$L)}
  \UnaryInfC{$\UniHyperMSetA \VL \UniSequent{(\AntOmSetA\cup\{\UniFmA\fus \UniFmB\};\UniMSetFmA)}{\UniMSetSucA}$}
  \DisplayProof
&
\AxiomC{$h_c \VL \UniSequent{(\AntOmSetA\cup\{\UniFmA_1\land \UniFmA_2\};\UniMSetFmA),\UniFmA_i}{\UniMSetSucA}$}
  \RightLabel{($\land$L)}
  \UnaryInfC{$\UniHyperMSetA \VL \UniSequent{(\AntOmSetA\cup\{\UniFmA_1\land \UniFmA_2\};\UniMSetFmA)}{\UniMSetSucA}$}
  \DisplayProof
\\[1.2em]
\multicolumn{3}{c}{
\begin{tabular}{cc}
  \AxiomC{$h_c \VL \UniSequent{(\AntOmSetA\cup\{\UniFmA\lor\UniFmB\};\UniMSetFmA),\UniFmA}{\UniMSetSucA}$}
  \AxiomC{$h_c \VL \UniSequent{(\AntOmSetA\cup\{\UniFmA\lor\UniFmB\};\UniMSetFmA),\UniFmB}{\UniMSetSucA}$}
  \RightLabel{($\lor$L)}
  \BinaryInfC{$\UniHyperMSetA \VL \UniSequent{(\AntOmSetA\cup\{\UniFmA\lor\UniFmB\};\UniMSetFmA)}{\UniMSetSucA}$}
  \DisplayProof
  &
  \AxiomC{$h_c \VL \UniSequent{(\AntOmSetA\cup\{\UniFmA\imp\UniFmB\};\UniMSetFmA)}{\UniFmA}$}
  \AxiomC{$h_c \VL \UniSequent{(\AntOmSetA\cup\{\UniFmA\imp\UniFmB\};\UniMSetFmB),\UniFmB}{\UniMSetSucA}$}
  \RightLabel{($\imp$L)}
  \BinaryInfC{$\UniHyperMSetA \VL \UniSequent{(\AntOmSetA\cup\{\UniFmA\imp\UniFmB\};\UniMSetFmB,\UniMSetFmA)}{\UniMSetSucA}$}
  \DisplayProof
\end{tabular}
}
\end{tabular}

\par\bigskip
\textbf{Logical rules operating on finite multiplicities}\par\medskip
\begin{tabular}{ccc}
  \AxiomC{$h_c \VL \UniSequent{(\AntOmSetA;\UniMSetFmA)}{\UniMSetSucA}$}
  \RightLabel{($1$L)$_2$}
  \UnaryInfC{$\UniHyperMSetA \VL \UniSequent{(\AntOmSetA;\UniMSetFmA,1)}{\UniMSetSucA}$}
  \DisplayProof
  &
  \AxiomC{$h_c \VL \UniSequent{(\AntOmSetA;\UniMSetFmA),\UniFmA,\UniFmB}{\UniMSetSucA}$}
  \RightLabel{($\fus$L)$_2$}
  \UnaryInfC{$\UniHyperMSetA \VL \UniSequent{(\AntOmSetA;\UniMSetFmA,\UniFmA\fus \UniFmB)}{\UniMSetSucA}$}
  \DisplayProof
&
\AxiomC{$h_c \VL \UniSequent{(\AntOmSetA;\UniMSetFmA),\UniFmA_i}{\UniMSetSucA}$}
  \RightLabel{($\land$L)$_2$}
  \UnaryInfC{$\UniHyperMSetA \VL \UniSequent{(\AntOmSetA;\UniMSetFmA,\UniFmA_1\land\UniFmA_2)}{\UniMSetSucA}$}
  \DisplayProof
\\[1.2em]
\multicolumn{3}{c}{
\begin{tabular}{cc}
  \AxiomC{$h_c \VL \UniSequent{(\AntOmSetA;\UniMSetFmA),\UniFmA}{\UniMSetSucA}$}
  \AxiomC{$h_c \VL \UniSequent{(\AntOmSetA;\UniMSetFmA),\UniFmB}{\UniMSetSucA}$}
  \RightLabel{($\lor$L)$_2$}
  \BinaryInfC{$\UniHyperMSetA \VL \UniSequent{(\AntOmSetA;\UniMSetFmA,\UniFmA\lor\UniFmB)}{\UniMSetSucA}$}
  \DisplayProof
&
  \AxiomC{$h_c \VL \UniSequent{(\AntOmSetA;\UniMSetFmA)}{\UniFmA}$}
  \AxiomC{$h_c \VL \UniSequent{(\AntOmSetA;\UniMSetFmB),\UniFmB}{\UniMSetSucA}$}
  \RightLabel{($\imp$L)$_2$}
  \BinaryInfC{$\UniHyperMSetA \VL \UniSequent{(\AntOmSetA;\UniMSetFmA,\UniMSetFmB,\UniFmA\imp\UniFmB)}{\UniMSetSucA}$}
  \DisplayProof
\end{tabular}
}
\end{tabular}

\par\bigskip
\textbf{Logical rules (succedent)}\par\medskip
\begin{tabular}{ccc}
  \AxiomC{$h_c \VL \UniSequent{(\AntOmSetA,\UniMSetFmA)}{}$}
  \RightLabel{($0$R)}
  \UnaryInfC{$\UniHyperMSetA \VL \UniSequent{(\AntOmSetA,\UniMSetFmA)}{0}$}
  \DisplayProof
  &
  \AxiomC{$h_c \VL \UniSequent{(\AntOmSetA;\UniMSetFmA)}{\UniFmA}$}
  \AxiomC{$h_c \VL \UniSequent{(\AntOmSetA;\UniMSetFmB)}{\UniFmB}$}
  \RightLabel{($\fus$R)}
  \BinaryInfC{$\UniHyperMSetA \VL \UniSequent{(\AntOmSetA;\UniMSetFmA,\UniMSetFmB)}{\UniFmA\fus \UniFmB}$}
  \DisplayProof
  &
  \AxiomC{$h_c \VL \UniSequent{(\AntOmSetA;\UniMSetFmA)}{\UniFmA_i}$}
  \RightLabel{($\lor$R)}
  \UnaryInfC{$\UniHyperMSetA \VL \UniSequent{(\AntOmSetA;\UniMSetFmA)}{\UniFmA_1\lor\UniFmA_2}$}
  \DisplayProof
  \\[1.2em]
  \multicolumn{3}{c}{
  \begin{tabular}{cc}
  \AxiomC{$h_c \VL \UniSequent{(\AntOmSetA;\UniMSetFmA)}{\UniFmA}$}
  \AxiomC{$h_c \VL \UniSequent{(\AntOmSetA;\UniMSetFmA)}{\UniFmB}$}
  \RightLabel{($\land$R)}
  \BinaryInfC{$\UniHyperMSetA \VL \UniSequent{(\AntOmSetA;\UniMSetFmA)}{\UniFmA\land \UniFmB}$}
  \DisplayProof
  &
  \AxiomC{$h_c \VL \UniSequent{(\AntOmSetA;\UniMSetFmA),\UniFmA}{\UniFmB}$}
  \RightLabel{($\imp$R)}
  \UnaryInfC{$\UniHyperMSetA \VL \UniSequent{(\AntOmSetA;\UniMSetFmA)}{\UniFmA \imp \UniFmB}$}
  \DisplayProof
  \end{tabular}
  }
  \hspace{2em}
\end{tabular}

\caption{
The calculus $\UniHFLewRinv_\omega$. 
Here, $h_c$ denotes the meta-hypersequent in the conclusion.}
\label{fig-omega-rules}
\end{figure*}


There is an obvious bijective correspondence between a sequent $\AntVecA \Ra b$ and the $\omega$-sequent $(\varnothing;\AntVecA)\Ra b$. 
This bijection extends from rule instances in $\UniHFLewRinv$ to rule instances in $\UniHFLewRinv_\omega$ that operate only on the finite coordinates of the $\omega$-hypersequent. Since the latter calculus does not contain any rule to strictly increase the $\omega$-coordinates from the empty set (we will subsequently introduce such a mechanism), the following is immediate.

\begin{lemmarep}\label{lem:original-calc-omega-calc}
Let~$h$ be an 
$\omega$-free
$\omega$-hypersequent.
Then, $h$ is provable in $\UniHFLewRinv$
    iff
    it is provable
    in $\UniHFLewRinv_\omega$.
\end{lemmarep}
\begin{proof}
    The left-to-right direction is 
    obvious, since
    $\UniHFLewRinv_\omega$ has the
    rules that operate on the second component.
    The right-to-left follows by an induction that relies on the observation that, in any rule instance, if the conclusion hypersequent is $\omega$-free, then so are all the premises; thus, by the mentioned bijection, there is a matching rule instance in $\UniHFLewRinv$.
\end{proof}

\subsection{Refined $\omega$-calculus}

Next, we modify $\UniHFLewRinv_\omega$ to equip it with height-preserving admissibility of $\UniEC$ and $\UniEW$, as in Sec~\ref{sec:refined-calculi-contraction}, based on~\cite{revantha2020}.

\begin{definition}[Refined $\omega$-calculus]
$\UniHCalcAbsorb{\UniHFLewRinv_\omega}$ consists of every rule $\UniHRuleAbsorb{\UniRuleSchemaA}$ below left s.t. $\UniRuleSchemaA$ (below right) is a rule in $\UniHFLewRinv_\omega$, and 
$g \UniExtCtrRel{\UniActiveCompFixed}{\UniEC} h$.
Call $\UniRuleSchemaA$ the \emph{base instance} of $\UniHRuleAbsorb{\UniRuleSchemaA}$.
\begin{center}
\begin{tabular}{c@{\hspace{2em}}c}
\AxiomC{$h \VL s_1$}
\AxiomC{$\cdots$}
\AxiomC{$h \VL s_n$}
\RightLabel{$\UniHRuleAbsorb{\UniRuleSchemaA}$}
\TrinaryInfC{
$h$
}
\DisplayProof
&
\AxiomC{$g \VL s_1$}
\AxiomC{$\cdots$}
\AxiomC{$g \VL s_n$}
\RightLabel{$\UniRuleSchemaA$}
\TrinaryInfC{$g$}
\DisplayProof
\end{tabular}
\end{center}
Every component in $h$ that is related, via applications of $\UniEC$, to the principal components in $g$ is called a principal component of $\UniHRuleAbsorb{\UniRuleSchemaA}$.
The \emph{schema} for an active/principal component
in $\UniHRuleAbsorb{\UniRuleSchemaA}$
is the corresponding schematic component in $\UniRuleSchemaA$.
\end{definition}

Informally, a rule instance in $\UniHCalcAbsorb{\UniHFLewRinv_\omega}$ (seen from conclusion to premises) is obtained by applying a limited amount---up to $\UniActiveCompFixed$, whose value depends solely on the proof rules---of $\UniEC$ backwards, and then a backward application of a rule instance in $\UniHFLewRinv_\omega$,
replacing $g$ with $h$ in the premises of the base instance.

\begin{lemmarep}\label{lem:omega-calc-omega-refined}
    The calculi
    $\UniHFLewRinv_\omega$
and $\UniHCalcAbsorb{\UniHFLewRinv_\omega}$ prove same $\omega$-hypersequents.
\end{lemmarep}
\begin{proof}
    Every $\omega$-hypersequent 
    provable in
    $\UniHFLewRinv_\omega$
    is provable in 
    $\UniHCalcAbsorb{\UniHFLewRinv_\omega}$
    because every rule instance of the former belongs in the latter.
Reverse direction: suppose that $h$ has a proof $\pi$ in
    $\UniHCalcAbsorb{\UniHFLewRinv_\omega}$.
    Structural induction on $\pi$.
    (\textit{Base case}) Immediate: apply $\UniEC$ to the base initial hypersequent.
    (\textit{Inductive case}) Suppose that $\pi$ concludes as below left. By definition of $\UniHCalcAbsorb{\UniRuleSchemaA}$ in $\UniHCalcAbsorb{\UniHFLewRinv_\omega}$, there is a base instance $\UniRuleSchemaA$ (below right) such that $g \UniExtCtrRel{\UniActiveCompFixed}{\text{(EC)}} h$.
\begin{center}
\begin{tabular}{c@{\hspace{2em}}c}
\AxiomC{$\pi_1$}
    \noLine
    \UnaryInfC{$h \VL s_1$}
    \AxiomC{$\ldots$}
    \AxiomC{$\pi_n$}
    \noLine
    \UnaryInfC{$h \VL s_n$}
    \RightLabel{$\UniHRuleAbsorb{\UniRuleSchemaA}$}
    \TrinaryInfC{$h$}
    \DisplayProof
&    
\AxiomC{$g \VL s_1$}
\AxiomC{$\cdots$}
\AxiomC{$g \VL s_n$}
\RightLabel{$\UniRuleSchemaA$}
\TrinaryInfC{$g$}
\DisplayProof
\end{tabular}
\end{center}
By IH,
each $h|s_i$ is provable in
$\UniHFLewRinv_\omega$. 
By $\UniEW$ we have 
$g|s_i$ provable in $\UniHFLewRinv_\omega$
as well. Applying the base instance instance $\UniRuleSchemaA$, we obtain $g$.
Using $\UniEC$, we reach a proof of $h$ in $\UniHFLewRinv_\omega$.

Final case: last rule is $\UniEC$ or $\UniEW$. It is not covered by above since rule form is slightly different. Nevertheless, the transformation is straightforward (apply same rules in $\UniHFLewRinv_\omega$).
\end{proof}

Similar to~\cite{revantha2020}, we have that the external structural rules are hp-admissible
in $\UniHCalcAbsorb{\UniHFLewRinv_\omega}$.

\begin{lemmarep}\label{lem:hp-admiss-ec-weakening}
 $\UniEW$
 and $\UniEC$
 are hp-admissible in
 $\UniHCalcAbsorb{\UniHFLewRinv_\omega}$.
\end{lemmarep}
\begin{proof}
    The proof is by induction on the height $\alpha$ of a proof $\pi$
    of an $\omega$-hypersequent $h$.

    $\UniEW$:
    The base case is obvious.
    Inductive step,
    assume $\pi$ has the following form, where each $\pi_i$ has height $\alpha_i < \alpha$.
We seek $h|t$.    
    \begin{center}
    \AxiomC{$\pi_1$}
    \noLine
    \UnaryInfC{$h \VL s_1$}
    \AxiomC{$\ldots$}
    \AxiomC{$\pi_n$}
    \noLine
    \UnaryInfC{$h \VL s_n$}
    \RightLabel{$\UniHRuleAbsorb{\UniRuleSchemaA}$}
    \TrinaryInfC{$h$}
    \DisplayProof
\end{center}
By IH,
the hypersequents
$(h|t|s_1), \ldots, (h|t|s_n)$
have proofs $\pi'_i$
of height $\leq \alpha_i < \alpha$.
Let the base instance $\UniRuleSchemaA$ of $\UniHRuleAbsorb{\UniRuleSchemaA}$ be as below left.
Then we have $g \UniExtCtrRel{\UniActiveCompFixed}{\text{(EC)}} h$.
By inspection, below right is a rule instance of $\UniHFLewRinv_\omega$.
\begin{center}
\begin{tabular}{c@{\hspace{1em}}c}
\AxiomC{$g \VL s_1$}
\AxiomC{$\cdots$}
\AxiomC{$g \VL s_n$}
\RightLabel{$\UniRuleSchemaA$}
\TrinaryInfC{$g$}
\DisplayProof
&
\AxiomC{$g \VL t \VL s_1$}
\AxiomC{$\cdots$}
\AxiomC{$g \VL t \VL s_n$}
\RightLabel{$\UniRuleSchemaA$}
\TrinaryInfC{$g \VL t$}
\DisplayProof
\end{tabular}
\end{center}
It suffices now to observe that $g\VL t \UniExtCtrRel{\UniActiveCompFixed}{\text{(EC)}} h\VL t$.

The last cases to consider is if $h$
resulted from $\UniEC$
or $\UniEW$, but that follows easily by the IH.


    $\UniEC$:
    The base case is obvious.
    Induction step:
    assume that $\pi$ has the following form, where each $\pi_i$ has height $\alpha_i < \alpha$.
    We seek $h\VL t$.

    \begin{center}
    \AxiomC{$\pi_1$}
    \noLine
    \UnaryInfC{$h\VL t\VL t \VL s_1$}
    \AxiomC{$\ldots$}
    \AxiomC{$\pi_n$}
    \noLine
    \UnaryInfC{$h\VL t\VL t \VL s_n$}
    \RightLabel{$\UniHRuleAbsorb{\UniRuleSchemaA}$}
    \TrinaryInfC{$h\VL t\VL t$}
    \DisplayProof
\end{center}
By IH,
$(h\VL t\VL s_1), \ldots, (h\VL t\VL s_n)$
have proofs $\pi'_i$
of height $\leq \alpha_i < \alpha$.
Now it is enough to prove that the calculus has the rule
instance $(h\VL t\VL s_1), \ldots, (h\VL t\VL s_n)/(h\VL t)$.
Observe that 
$\UniHFLewRinv_\omega$
contains the base rule instance $\UniRuleSchemaA$ below such that 
$g \UniExtCtrRel{\UniActiveCompFixed}{\text{(EC)}} h\VL t\VL t$.
\begin{center}
\AxiomC{$g \VL s_1$}
\AxiomC{$\cdots$}
\AxiomC{$g \VL s_n$}
\RightLabel{$\UniRuleSchemaA$}
\TrinaryInfC{$g$}
\DisplayProof
\end{center}
There are two possibilities.
\begin{enumerate}
    \item If $|g|_t - |(h \VL t)|_t < \UniActiveCompFixed$ i.e., we have not exhausted the permitted number of $\UniEC$ rules, so we can apply one more i.e., $g \UniExtCtrRel{\UniActiveCompFixed}{\text{(EC)}} h\VL t$.
    
    \item Else $|g|_t - |(h \VL t)|_t = \UniActiveCompFixed$. By choice of $\UniActiveCompFixed$, at least one copy of $t$ must have been instantiated as part of the hypersequent variable $H$ in the rule schema of $\UniRuleSchemaA$ (pigeon-hole principle). Consider variant $\UniRuleSchemaA'$
    of $\UniRuleSchemaA$
    where this copy of $t$
    is omitted from the instantiation of $H$. Using $\UniHRuleAbsorb{\UniRuleSchemaA'}$ in place of $\UniHRuleAbsorb{\UniRuleSchemaA}$ yields $h\VL t$.
\end{enumerate}
The last cases are $\UniEC$, which follows by IH, and $\UniEW$, which follows from hp-admissibility of $\UniEW$ proved above.
    \qedhere
\end{proof}



\begin{lemma}
    If $h$ is provable in $\UniHCalcAbsorb{\UniHFLewRinv_\omega}$, then it is provable without $\UniEW$ and $\UniEC$.
\end{lemma}

Henceforth, we may assume that proofs do not use these rules.



The rules operating on $\omega$-coordinates behave just like their familiar counterparts operating on the finite multiplicities (except $\omega$-coordinates do not diminish with use).
Hence, treating a finite multiplicity formula as belonging to the $\omega$-set is proof-preserving:

\begin{definition}[$\omega$-extensions]\label{def:omega-extensions}
Given an $(\omega,d)$-sequent $s = \UniSequent{(\AntOmSetA,\AntVecA)}{b}$, we say that
an $(\omega,d)$-sequent $s'=\UniSequent{(\AntOmSetA',\AntVecA')}{b}$ is an \emph{$\omega$-extension} of $s$ if  $\AntOmSetA \subseteq \AntOmSetA'$
and $\AntVecA' = \AntVecA[\AntOmSetA'{\setminus}\AntOmSetA \mapsto 0]$.
An \emph{$\omega$-extension of $h$} is obtained by replacing some components in the $(\omega,d)$-hypersequent $h$ by its $\omega$-extensions.
\end{definition}

\begin{lemma}\label{lem:hp-admin-omega-function}
Let $h$ be an $\omega$-hypersequent provable in $\UniHCalcAbsorb{\UniHFLewRinv_\omega}$ with height $\alpha$.
Then, any of its $\omega$-extensions is provable
in
$\UniHCalcAbsorb{\UniHFLewRinv_\omega}$
with height $\leq \alpha$.
\end{lemma}
\begin{proof}
    Let $h$ be an $\omega$-hypersequent having a proof
    $\pi$ of height $\alpha$ in $\UniHCalcAbsorb{\UniHFLewRinv_\omega}$.
    It is enough to prove
    for the case
    when $h'$
    is an $\omega$-extension
    of $h$ differing from $h$ at a single
    component $\UniSequent{(\AntOmSetA',\AntVecA')}{b}$, and, within that component, at a single formula $a$.
    That is, $h = g \VL \UniSequent{(\AntOmSetA,\AntVecA)}{b}$ with $a\not\in \AntOmSetA$
    and $h' = g \VL \UniSequent{(\AntOmSetA\cup\{a\},\AntVecA[a \mapsto 0])}{b}$.
    The general result follows by iterating this case.
    We will now show 
    by structural induction on $\pi$
    that $h'$
    has a proof $\pi'$ of height $\leq \alpha$.

    \textit{(Base case)}
    Here, $h$ is an initial $\omega$-hypersequent.
    If the component that changed in $h'$ is part of the instantiation of
    the hypersequent variable $H$, simply change the instantiation of $H$ accordingly.
    Otherwise, 
    it must be of
    one of the forms
    $\UniSequent{(\AntOmSetA;p)}{p}$
    or
    $\UniSequent{(\AntOmSetA;0)}{}$.
    In that case, use instead
    the initial hypersequents with
    principal components
    $\UniSequent{(\AntOmSetA \cup \{p\};\mathbf 0)}{p}$
    or
    $\UniSequent{(\AntOmSetA\cup\{0\};\mathbf 0)}{}$, respectively.
    
    \emph{(Inductive step)}
    Assume that
    the last rule instance applied in
    $\pi$
    has the form 
    $( h \VL t_i )_i / h$,
    where
    \[(g|s_1|\cdots|s_m|t_1),\ldots,(g|s_1|\cdots|s_m|t_k)/(g|s_1|\cdots|s_m)\]
    is its base rule instace in
    $\UniHFLewRinv_\omega$
    and $(g|s_1|\cdots|s_m) \UniExtCtrRel{\UniActiveCompFixed}{\UniEC} h$.
    Here, $g$ is the instantiation of the hypersequent variable $H$.
    
    By IH,
    we have $h' | t_i$ provable with height $\leq \alpha_i < \alpha$
    for each $i$.
    We show that the instance 
    $(h' | t_i)_i / h'$
    is in the calculus.
    By cases:
    \begin{enumerate}
        \item  The component
    that changed from $h$ to $h'$ is in $g$:
    use the IH and instantiate the same rule schema of $\UniHFLewRinv_\omega$
    but changing $g$
    to $g'$ in the instantiation
    of the hypersequent variable $H$
    to obtain $h'$.
    That is,
     \[
     \qquad\quad(g'|s_1|\cdots|s_m|t_1),\ldots,(g'|s_1|\cdots|s_m|t_k)/(g'|s_1|\cdots|s_m)\]
    is a rule instance of
    $\UniHFLewRinv_\omega$.
    So, 
    $g'|s_1|\cdots|s_m
\UniExtCtrRel{\UniActiveCompFixed}{\UniEC} h'$.
    \item 
    Assume that 
    the component that
    changed is $s_i$
    for some $1 \leq i \leq m$, and $s_i$ is not present in $g$ (otherwise the previous item would have applied).
    If the last rule applied 
    operates on the
    first component
    or on succedents,
    note that $i=1$ necessarily (there is only one principal component), so
    simply consider
    the component $s_1'$
    that now has the formula $a$ in the first component, and adapt the rule instance accordingly.
    In the case of rules operating on the first component, we might get in the new instantiation  premises that are $\omega$-extensions of the previous premises, so just apply the IH to obtain proofs of height $\leq \alpha_i < \alpha$ for them.

    If it is a logical rule operating on the second component or $(\UniWProp)$, again we have $i=1$,
    and now we need to
    consider two cases.
    First, the formula
    that changed in the component is not principal.
    In this case, simply use the IH and adapt the instantiation.
    If it is the principal formula
    of the rule,
    use the IH
    and then appropriately instantiate the corresponding logical rule 
    that operates on the first component.

    If it is a structural rule from $\UniAnaRuleSet_\omega$,
    the formula $a$
    must be in an instantiation of a multiset meta-variable.
    By the property of linear conclusion of these rules, this variable only occurs in the schema of $s_i$.
    By definition of $\UniRuleSchemaA_\omega$,
    all premises sharing same variable get affected, and their corresponding 
    $\omega$-sets are changed accordingly.
    Thus, it is enough to change the instantiation of that variable to not contain $a$, and place $a$ in $\AntOmSetA_i$.
    Thus use the IH and this new instantiation.
    \end{enumerate}

    The lemma then follows by iterating the above argument.
\end{proof}

\section{Proof search algorithm for extensions of $ \UniFLeExtLogic{\UniWProp}$}
\label{sec:ub-flew-proof-search}

Now we build on the refined $\omega$-calculus developed in the
previous section to describe a proof search procedure.

\begin{definition}\label{def:order-ant}
Define $\prec$
on $(\omega,d)$-sequents as
$s \prec t$
iff $\AntOmSetA \subseteq \AntOmSetA'$,
     $\AntVecA < \AntVecB$
     and $b_1=b_2$,
for any
 $s = \UniSequent{(\AntOmSetA,\AntVecA)}{b_1}$ and $t = \UniSequent{(\AntOmSetA',\AntVecB)}{b_2}$.
\end{definition}

The following function identifies strict increments in the finite multiplicities, and includes those formulas as $\omega$-coordinates.

\begin{definition}[$\omega$-function]
    Let $s = \UniSequent{(\AntOmSetA,\AntVecA)}{b}$ and $t = \UniSequent{(\AntOmSetA',\AntVecB)}{b}$
    be $\omega$-sequents
    such that 
    $s \prec t$.
    By $\omega(s, t)$
    we denote the $\omega$-sequent
    $\UniSequent{(\AntOmSetA'',\AntVecB')}{b}$, where
    $\AntOmSetA'' = \AntOmSetA' \cup \kappa$
    and $\AntVecB' = \AntVecB[\kappa \mapsto 0]$, with
    $\kappa = \{ l \in \mathbb{N}_{\leq d}^+{\setminus}\AntOmSetA' \mid \AntVecA[l] < \AntVecB[l] \}$.
\end{definition}

\begin{example}\label{ex:omega-function}
Let
$s_1 = \UniSequent{(1,3;(0,2,0,4))}{b}$,
$s_2 = (1,3,4;(0,3,0,0)) \allowbreak \Ra b$.
Then $\omega(s_1,s_2) = \UniSequent{(1,2,3,4;(0,0,0,0))}{b}$.
\end{example}

Proof search will need to track more than just the hypersequents that appear. So we first enrich each component $s$ of the hypersequent as $(i,s)$ where $i\in\mathbb{N}$ serves as a unique index recording the order in which the components were created. 

Let $d \geq 1$.
An \emph{indexed $d$-component}
is a pair $\MkIdx{s}=(i,s)$,
where $i \in \mathbb{N}$ and $s$ is an $(\omega,d)$-sequent. We denote by $\Cmp{\MkIdx{s}}$
the $(\omega,d)$-sequent of
$\MkIdx{s}$.
An \emph{indexed $d$-hypersequent} is a 
multiset of indexed $d$-components, such that any two distinct indexed components 
have distinct indices (we could use a set here, but multisets are familiar for hypersequents).
We omit $d$ from the notation when clear from context.
The \emph{support} of an indexed hypersequent is the hypersequent obtained by deleting the index of every component. 
If the support has no repeated components, $\MkIdx{h}$ is said to be \emph{irredundant}.

\begin{definition}[Relational indexed hypersequent]\label{def-rin}
For $d \geq 1$,
a
$d$-\emph{\rinhypersequent} is
a tuple $\mHyperA= \UniEnrichHyper{\IdxHypA}{\UniRelDep{\mHyperA}}$, where
$\IdxHypA$ is an
indexed
$d$-hypersequent,
and
$\UniRelDep{\mHyperA}$ is a binary relation
on the indexed components in $\IdxHypA$ such that $\UniRelDep{\mHyperA}^*$ is a rooted tree.
Let $\max \IdxHypA$ denote the largest index in $\IdxHypA$.
We say that $\mHyperA$
is \emph{irredundant} when
$\MkIdx{h}$ is so.
\end{definition}

Let
$\UniRuleSchemaA = (h \VL s_i)_i/h$ be a rule of $\UniHCalcAbsorb{\UniHFLewRinv_\omega}$ and
$\mHyperA= \UniEnrichHyper{\IdxHypA}{\UniRelDep{\mHyperA}}$
be an irredundant \rinhypersequent{}
with support $h$.
We refer to the structure $\mathfrak{I} = \langle \UniRuleSchemaA, \mHyperA \rangle$ as an \emph{instantiation context}.
Let $C$ be the multiset of
principal components of $h$
in $\UniRuleSchemaA$.
We say that an indexed component in $\MkIdx{h}$ is principal
in $\UniRuleSchemaA$ if its $\omega$-sequent is in $C$. The \emph{schema} of such component in $\UniRuleSchemaA$ is defined as the schema of the corresponding principal component in $C$.
Define the \emph{key ancestor 
of $s_i$ in $\mathfrak{I}$} as the 
principal component of
$\MkIdx{h}$ in
$\UniRuleSchemaA$
of largest index among those 
whose schema shares a variable with the schema of $s_i$.

If $\MkIdx{t}$ is the key ancestor
of $s_i$ in $\mathfrak{I}$,
define $\MkIdx{s}_i = (\max \IdxHypA + 1, s_i)$
and $R_i = \UniRelDep{\mHyperA} \cup \{ (\MkIdx{t}, \MkIdx{s}_i) \}$.
The \emph{$\omega$-partner
of $\MkIdx{s}_i$ in $\mathfrak{I}$}
(if it exists)
is the indexed component $\MkIdx{t}'$
in $\IdxHypA$
of largest index such that 
$(\MkIdx{t}', \MkIdx{s}_i) \in R_i^*$
and $\Cmp{\MkIdx{t}'} \prec s_i$.
Then, for $s'_i = \omega(t', s_i)$,
the indexed $\omega$-sequent $\MkIdx{s}'_i = (\max \IdxHypA + 1,s'_i)$ is called the \emph{$\omega$-refinement of $\MkIdx{s}_i$
in $\mathfrak{I}$}, and
$\MkIdx{s}_i$ is called the \emph{pre-component of $\MkIdx{s}'_i$ in $\mathfrak{I}$}.


\begin{definition}\label{def:lrin-hyper}
A \emph{legal \rinhypersequent{}}
(or \lrinhypersequent{}, for short)
is an irredundant \rinhypersequent{} 
$\mHyperA= \UniEnrichHyper{\IdxHypA}{\UniRelDep{\mHyperA}}$ such that
if $\MkIdx{s} \UniRelDep{\mHyperA} \MkIdx{t}$, then there is a rule $\UniRuleSchemaA$ such that, for $\mathfrak{I} = \langle \UniRuleSchemaA, \mHyperA \rangle$,
(a):
$\Cmp{\MkIdx{t}}$ is one of the active components of $\UniRuleSchemaA$ and 
$\MkIdx{s}$ is the key ancestor of $\Cmp{\MkIdx{t}}$ in $\mathfrak{I}$, and 
(b): $\MkIdx{t}$ either has no $\omega$-partner in $\mathfrak{I}$,
or is an $\omega$-refinement
in $\mathfrak{I}$.
\end{definition}



\begin{definition}[$\omega$-proof-search tree]\label{def:omega-proof-search}
For 
an \lrinhypersequent{}
$\mHyperA_0$,
the $\omega$-proof-search tree
$T_{\omega}(\mHyperA_0)$ is defined
as the limit of the trees recursively constructed in the following way.

\begin{enumerate}[leftmargin=*, label={}]
\item $T_0(\mHyperA_0)$: A single node labelled with 
$\mHyperA_0$.
\item $T_{n+1}(\mHyperA_0)$: Pick a leaf~$\ell$ of $T_n(\mHyperA_0)$
and assume it is labeled by $\mHyperA= \UniEnrichHyper{\IdxHypA}{\UniRelDep{\mHyperA}}$.
For any rule $\UniRuleSchemaA = (h|s_i)_{i \in I} / h$
of $\UniHCalcAbsorb{\UniHFLewRinv_\omega}$ resulting from an instantiation with subformulas of the support of $\mHyperA_0$, where $h$ is the support of $\IdxHypA$, we do the following.

\begin{enumerate}[leftmargin=*]
    \item (pre-refinement redundancy) If
    there is $i \in I$
    and $\MkIdx{t}$ in $\IdxHypA$
    s.t.
    $s_i = \Cmp{\MkIdx{t}}$, 
    then go to the next iteration,
    else proceed:
\end{enumerate}

Let $\mathfrak{I} \UniSymbDef \langle \UniRuleSchemaA, \mHyperA \rangle$,
and, for each $i \in I$, set $m \UniSymbDef \max \IdxHypA+1$,
$\MkIdx{s}_i \UniSymbDef (m,s_i)$ and $R_i \UniSymbDef \UniRelDep{\mHyperA} \cup \{ (\MkIdx{t}_i,\MkIdx{s}_i) \}$,
where $\MkIdx{t}_i$ is the key ancestor of $s_i$
in $\mathfrak{I}$.
Then perform the following.
\begin{enumerate}[label=(\alph*),start=2]
    \setcounter{enumi}{1}
    \item (Leaf expansion)
    For each $i \in I$,
    add
    $\UniEnrichHyper{\IdxHypA|\MkIdx{s}_i}{R_i}$ as a child of~$\ell$.
    \item ($\omega$-refinement)
    For each new
    leaf
    $\UniEnrichHyper{\IdxHypA|(m,s_i)}{R_i}$,
    let $\MkIdx{s}'_i$ be the \emph{$\omega$-refinement
of $\MkIdx{s}_i$ in $\mathfrak{I}$}
(if it exists), 
    and replace 
    the label of this leaf
    by $\UniEnrichHyper{\IdxHypA|\MkIdx{s}'_i}{\UniRelDep{\mHyperA} \cup \{ (\MkIdx{t}_i,\MkIdx{s}'_i) \}}$. 

    If this refinement took place for at least one premise, we call the result an
    \emph{$\omega$-introduction rule}, with $\MkIdx{s}'_i$ called a \emph{productive premise} and $r$ 
    its \emph{pre-rule}.
    \item (post-refinement redundancy)
    If the \rinhypersequent{} of any new leaf 
    is redundant, delete all the new leaves introduced in this iteration.
\end{enumerate}
\end{enumerate}
\end{definition}

Observe that the tree resulting
from the above construction
is guaranteed to have legal \rinhypersequent{}s as labels. 
Consider an \lrinhypersequent~$\UniEnrichHyper{\IdxHypA}{R}$ in $T_{\omega}(\mHyperA_0)$.
Viewing the indexed hypersequent $\IdxHypA$ as a sequence of indexed components, ordered by increasing index, we have access to the order in which the components were obtained. Also, $((i,t),(j,s))\in R$ tells us $t$ was a principal component of the rule instance that had $s$ as an active component, and 
the instantiation that led to $s$ is independent of components with index $>i$, in view of the key ancestor condition. 


\begin{definition}\label{def:tree-R}
Let
$\mHyperA \UniSymbDef \UniEnrichHyper{\IdxHypA}{\UniRelDep{\mHyperA}}$
be an \lrinhypersequent.
We define the \emph{tree of $\mHyperA$}, denoted by
$\mTree{\mHyperA}$,
by first making the
minimal elements of $\UniRelDepZero{\mHyperA}$
the children of a special
node labelled with $\bullet$,
and then extending it 
such that
$\MkIdx{s}'$ is a child of $\MkIdx{s}$ iff
$(\MkIdx{s},\MkIdx{s'}) \in \UniRelDep{\mHyperA}$.
\end{definition}

\begin{remark}
    Let
    $\mHyperA \UniSymbDef \UniEnrichHyper{\IdxHypA}{\UniRelDep{\mHyperA}}$
be an {\lrinhypersequent}
occurring in $T_\omega(\mHyperA_0)$. 
First, note
that $T_\omega(\mHyperA_0)$ is a tree labeled by {\lrinhypersequent}s, while
$\mTree{\mHyperA}$ is a 
tree labeled by components, plus a special symbol $\bullet$, which labels the root.
Also,
the children of $\bullet$
are elements of $\IdxHypA_0$, so the tree is finitely branching at the root (since $\IdxHypA_0$ is finite).
\end{remark}

We now define a more restrictive wqo than $\prec$ (Def~\ref{def:order-ant})
where indexed $(\omega,d)$-sequents with distinct $\omega$-sets are incomparable.

\begin{definition}
Given $\MkIdx{s} = (i, \UniSequent{(\AntOmSetA;\AntVecA)}{b_1})$
and 
$\MkIdx{s'} = (j, \UniSequent{(\AntOmSetA';\AntVecB)}{b_2})$,
let $\MkIdx{s} \preceq_{\mathbf{W}(\omega,d)} \MkIdx{s'}$
iff $\AntOmSetA = \AntOmSetA'$,
$b_1 = b_2$
and
$\AntVecA \leq_{\mathbb{N}^d} \AntVecB$.
Define the structure $\mathbf{W}(\omega,d)$
on indexed $(\omega,d)$-sequents
having
$\preceq_{\mathbf{W}(\omega,d)}$ as relation and
$\UniNorm{(l,\MkIdx{s})}{} \UniSymbDef \UniNorm{\Cmp{\MkIdx{s}}}{}$
as norm, where the norm of an $\omega$-sequent
$\UniSequent{(\AntOmSetA;\AntVecA)}{b}$ is defined
as $\UniNorm{\AntVecA}{\mathbb{N}^d}$.
\end{definition}

\begin{lemma}\label{lem:nwqo-weakening}
    $\mathbf{W}(\omega,d)$
    is an nwqo
    and, for any primitive recursive control function $\UniControlFunctionA$,
    its length function 
    $L_{\mathbf{W}(\omega,d),\UniControlFunctionA}$
    is upper bounded in
    $\UniFGHOneAppLevel{\omega}$, uniformly on $d$.
\end{lemma}
\begin{proof}
In order to see that
$\mathbf{W}(\omega,d)$ is a wqo,
let $\UniWqoA = (2^d \cdot (d+1)) \cdot \UniWqoModNatural^{d}$.
Now, uniquely assign a number $I(\AntOmSetA,b)$ in 
$\mathbb{N}_{\leq 2^d \cdot (d+1)}$
to the pairs $(\AntOmSetA,b)$, where
$\AntOmSetA$ is a subset
of $\mathbb{N}_{\leq d}^+$
and $b \in \mathbb{N}_{\leq d}$.
We then code the indexed $(\omega,d)$-sequent
    $\MkIdx{s} = (i,\UniSequent{(\AntOmSetA,\AntVecA)}{b})$
as the element 
$s^\# \UniSymbDef (I(\AntOmSetA,b),\AntVecA)$
of $\UniWqoA$ (note that the index is ignored).
Then clearly
$\MkIdx{t} \preceq_{\mathbf{W}(\omega,d)} \MkIdx{s}$
iff $\MkIdx{t}^\# \leq_{\UniWqoA} \MkIdx{s}^\#$. 
Moreover,
$\UniNorm{\MkIdx{t}^\#}{}
= \UniNorm{\MkIdx{t}}{}$.
It then follows that any $(\UniControlFunctionA,n)$-controlled bad sequence $\MkIdx{t_0},\MkIdx{t_1},\ldots,$ over sequents
can be translated to a $(\UniControlFunctionA,n)$-controlled bad sequence $\MkIdx{t}_0^\#,\MkIdx{t}_1^\#,\ldots$
Hence, the 
 length functions of
 $\mathbf{W}(\omega,d)$
  are upper bounded by the length functions of 
  $\UniWqoA$. We conclude by applying Thm~\ref{the:length-theorem-dickson}.
\end{proof}



\begin{lemma}\label{lem:omega-coordinate-pres}
Let
$\mHyperA_0 \UniSymbDef \UniEnrichHyper{\IdxHypA_0}{\UniRelDep{\mHyperA_0}}$ and
$\mHyperA \UniSymbDef \UniEnrichHyper{\IdxHypA}{\UniRelDep{\mHyperA}}$
be \lrinhypersequent{}s
such that 
$\mHyperA_0$ is finite
and
$\mHyperA$ occurs in $T_\omega(\mHyperA_0)$.
Then:
\begin{enumerate}
    \item
    Along a branch of 
    $\mTree{\mHyperA}$,
    indices of components strictly increase and
$\omega$-sets are 
    increasing.
        \item 
    Every branch of $\mTree{\mHyperA}$
    is an $(f, \UniNorm{h_0}{})$-controlled bad sequence over 
    $\mathbf{W}(\omega,d)$,
    where $f$ is the function from Lem~\ref{lem:control-hyper-calc}.
    \item $\mTree{\mHyperA}$ is finitely branching.
\end{enumerate}
Therefore,
$\mTree{\mHyperA}$ (and hence $\mHyperA$) is finite.
There is a function in $\UniFGHOneAppLevel{\omega}$ that uniformly upper bounds the size of $\mTree{\mHyperA}$
as
    a function of $\UniNorm{\Cmp{\MkIdx{h}_0}}{}$.
\end{lemma}
\begin{proof}
Recall that
a branch of $\mTree{\mHyperA}$ is---after $\bullet$---a chain of 
indexed components $\MkIdx{t}_1
\,\UniRelDepZero{\mHyperA}\,
\MkIdx{t}_2
\,\UniRelDepZero{\mHyperA} \ldots$
such that $\MkIdx{t}_{i+1}$
is the active component of a premise of a rule
of the calculus
having $\MkIdx{t}_i$
as a principal component
(in fact, as its key ancestor).

\underline{Item 1.} Indices of components strictly increase because, by construction, $\MkIdx{t}_i \UniRelDepZero{\mHyperA}
\MkIdx{t}_{i+1}$ implies that the index of the component $\MkIdx{t}_{i+1}$ is strictly larger than the index of any of its ancestors.
For the next part, by the way the rules are defined, an active component has a larger $\omega$-set than its key ancestor for any non $\omega$-introduction rule. If an active component was created by an $\omega$-introduction rule, then it has an $\omega$-set that is strictly larger than its parent (and in fact, all its ancestors, by an implicit induction).

\underline{Item 2.} 
Let
$\mathsf{B} \UniSymbDef \MkIdx{t}_1
\,\UniRelDepZero{\mHyperA}\,
\MkIdx{t}_2
\,\UniRelDepZero{\mHyperA} \ldots$
be a branch of
of $\mTree{\mHyperA}$.
If it is not a bad sequence,
there is a chain 
$\UniSequent{(\AntOmSetA_i,\AntVecA_i)}{b_i} =\MkIdx{t}_i \allowbreak \,\UniRelDepZero{\mHyperA}\,
\ldots 
\,\UniRelDepZero{\mHyperA}\,
\MkIdx{t}_j=
\UniSequent{(\AntOmSetA_j,\AntVecA_j)}{b_j}$ in $\mathsf{B}$ for
some $i<j$ such that 
$\MkIdx{t}_i\preceq_{\mathbf{W}(\omega,d)}\MkIdx{t}_j$ i.e., $\AntOmSetA_i= \AntOmSetA_j$, $b_i=b_j$ and $\AntVecA_i < \AntVecA_j$ (strict since $\mHyperA$ is irredundant).  
Notice that $\MkIdx{t}_j$ must be the result of $\omega$-refinement since otherwise it would mean that the algorithm missed a suitable $\omega$-partner (i.e., $\MkIdx{t}_i$). 
Denote the pre-component of $\MkIdx{t}_j$
by $\MkIdx{t}'_j = \UniSequent{(\AntOmSetA_j',\AntVecB_j')}{b}$.
Then $\AntOmSetA_j'\subset \AntOmSetA_j$ since, as observed above, the pre-component has a strictly smaller $\omega$-set than the productive premise that replaces it. Since $j > i$, and since $\omega$-sets are increasing (by $\omega$-increasing property of related components from conclusion to premise), it follows that $\AntOmSetA_i\subseteq\AntOmSetA_j'$. 
Since $\AntOmSetA_i= \AntOmSetA_j$, we get the contradiction $\AntOmSetA_i\subseteq\AntOmSetA_j'\subset \AntOmSetA_i$.

    Moreover, 
    since $(\MkIdx{t}_i,\MkIdx{t}_{i+1}) \in \UniRelDepZero{\mHyperA}$ means that 
    $\MkIdx{t}_{i+1}$
    is either a premise or an $\omega$-extension of a premise
    of a rule with 
    $\MkIdx{t}_i$
    among its principal components, 
    we have that $\UniNorm{\MkIdx{t}_{i+1}}{} \leq f(\UniNorm{\MkIdx{t}_{i}}{})$
    by Lem~\ref{lem:control-hyper-calc},
    and thus
    $\UniNorm{\MkIdx{t}_{j}}{}
    \leq f^{j}(\UniNorm{\MkIdx{t}_{1}}{})$, meaning that the sequence is ($f,\UniNorm{h_0}{}$)-controlled, since
    $\UniNorm{\MkIdx{t}_1}{} 
    \leq \UniNorm{h_0}{}$
    because $\MkIdx{t}_1$ is in $\MkIdx{h}_0$.
    \underline{Item 3.} 
    Towards a contradiction, 
    pick a node
    labelled by 
    $\MkIdx{t}$
    having infinitely-many children, i.e.,
    $\MkIdx{t} \,\UniRelDepZero{\mHyperA} \, \MkIdx{s}_i$
    for all $i \geq 0$.
So $\MkIdx{t}$
acted as key ancestor
for each $\MkIdx{s}_i$, 
meaning that it was the most recent
principal component in the 
indexed \rinhypersequent{}
$\mHyperA_i \UniSymbDef \langle \IdxHypA_i,R_i \rangle$
being expanded 
sharing a multiset meta-variable with the schema of $\MkIdx{s}_i$.

Observe that the set~$A$ of
components with smaller
index than $\MkIdx{t}$
in the whole $\IdxHypA$
is the same for 
all such $\mHyperA_{i}$.
This means that every $\MkIdx{s}_i$ is determined by $A$. It follows that $\MkIdx{s}_i=\MkIdx{s}_j$ for some $i,j$, which contradicts the redundancy check.
Observe also that there are only finitely-many 
children of the first special node $\bullet$, since $\MkIdx{h}_0$ is 
finite. This finishes the proof of Item (3).


So $\mTree{\mHyperA}$ is finitely branching and all of its branches are finite; thus, by K\"onig's Lemma, 
$\mTree{\mHyperA}$ is finite. Now let $m$ be the length of the longest branch in this tree.
By the argument above, at any node labelled by $\MkIdx{t}$,
there can only be $p(\UniNorm{\MkIdx{t}}{}) \le p(f^m(\UniNorm{h_0}{}))$ many children of it for some primitive recursive function $p$ that depends only on the rules of the underlying calculus. 
Hence, the size of the total tree (which is at most the maximum number of children a node can have exponentiated to the length of the longest branch) is at most
$(p(f^m(\UniNorm{h_0}{}))^m$. This is a primitive recursive function in $m$ (\cite[Definition 9]{Clote94},\cite[Section 5.3.1]{schmitz2016hierar}).
Since $m$ is upper bounded by a function in $\UniFGHOneAppLevel{\omega}$
by Item (2) and Lem~\ref{lem:nwqo-weakening},
then by Thm~\ref{thm:composition-fomega} it follows that the size of the whole tree is bounded by a function in $\UniFGHOneAppLevel{\omega}$.
\qedhere

\end{proof}


\begin{theorem}\label{the:termination-proof-search}
Let $\mHyperA_0$ be a finite \lrinhypersequent{}.
Then
$T_\omega(\mHyperA_0)$
is finite, and thus the construction of $T_\omega(\mHyperA_0)$ terminates.
Moreover,
$T_{\omega}(\mHyperA_0)$
is built in time $\UniFGHOneAppLevel{\omega}$.
\end{theorem}
\begin{proof}
Assume that $T_\omega(\mHyperA_0)$ is
not finite.
As branching factor of proof search tree is finite (the calculus is locally finitary), $T_{\omega}(\mHyperA_0)$ must have an infinite branch
by K\"onig's Lemma, 
which induces an infinite \lrinhypersequent{} $\mHyperA \UniSymbDef \UniEnrichHyper{\IdxHypA}{\UniRelDepZero{\mHyperA}}$, contradicting
Lem~\ref{lem:omega-coordinate-pres}, which says that there cannot be infinite \lrinhypersequent{}
in $T_\omega(\mHyperA_0)$.

A uniform $\UniFGHOneAppLevel{\omega}$ upper bound for the length
    of any branch in $T_{\omega}(\mHyperA_0)$
    is given by Lem~\ref{lem:omega-coordinate-pres}.
    The branching factor is also  $\UniFGHOneAppLevel{\omega}$ upper bounded.
    So the size of this tree is upper bounded in $\UniFGHOneAppLevel{\omega}$.
    Because performing instantiation, expanding nodes, 
    and doing all the required checks in the construction of $T_\omega(\mHyperA_0)$
    are all primitive recursive operations
    over an object of size $\UniFGHOneAppLevel{\omega}$,
    the whole tree is built in this time.\qedhere
\end{proof}


\begin{definition}[$\omega$-eager proofs]
A subtree of $T_{\omega}(\mHyperA_0)$ 
rooted at $\mHyperA_0$ 
whose leaves were obtained from 
instances of axiomatic rules 
is called an \emph{$\omega$-eager proof of $h_0$}, where $h_0$
is the support of $\IdxHypA_0$.
\end{definition}

\newcommand{\UniPSWknAlg}[1]{\mathtt{Wkn}(#1)}
\newcommand{\UniPSWknAlgnoinput}{\mathtt{Wkn}}

To any irredundant $\omega$-hypersequent $h_0$
we canonically associate
the \lrinhypersequent{}
$\mHyperA_0^{\circ} = \UniEnrichHyper{\IdxHypA_0}{\varnothing}$,
where $\IdxHypA_0$ is obtained
by choosing an order for the
elements of $h_0$ and assigning
indices $\{0,\ldots,|h_0|-1\}$
respecting this order.

\begin{definition}[$\omega$-proof search algorithm]
Given a finite irredundant $\omega$-hypersequent $h_0$ 
as input, 
the algorithm $\UniPSWknAlg{h_0}$
builds the tree
$T_\omega(\mHyperA_0^{\circ})$,
then looks for a subtree
that is an $\omega$-eager proof
of $h_0$.
If found, it outputs {\normalfont\texttt{provable}},
otherwise it outputs {\normalfont\texttt{unprovable}}.
\end{definition}



By 
Lem~\ref{lem:omega-coordinate-pres} and
Lem~\ref{the:termination-proof-search} we obtain termination and complexity.

\begin{theoremrep}[complexity]\label{the:complexity-wkn}
    $\UniPSWknAlgnoinput$
    terminates in $\UniFGHOneAppLevel{\omega}$ time.
\end{theoremrep}
\begin{proof}
    The algorithm obviously terminates in view of Lem~\ref{the:termination-proof-search}.
    Recall that $\UniPSWknAlg{h_0}$
    builds $T_\omega(\mHyperA_0^{\circ})$ in its first step, and 
    by Lem.~\ref{lem:omega-coordinate-pres},
    the size of this tree is upper bounded by a $\UniFGHOneAppLevel{\omega}$ function.
    Performing instantiation, expanding nodes, 
    and doing all the required checks in the construction of $T_\omega(\mHyperA_0^{\circ})$,
and checking the subtrees of $\omega$-eager proofs are operations that are primitive recursive in the size of the object. Hence, by Theorem~\ref{thm:composition-fomega}, the algorithm
    runs in time that is primitive recursive in a $\UniFGHOneAppLevel{\omega}$ function i.e., in $\UniFGHOneAppLevel{\omega}$.
\end{proof}

\section{Correctness of the $\omega$-proof search algorithm for extensions of $ \UniFLeExtLogic{\UniWProp}$}
\label{sec:ub-flew-correctness}
Given a finite irredundant hypersequent $h_0$,
we say that \emph{$h_0$ has an $\omega$-eager proof} if
$\UniPSWknAlg{h_0}$ returns \texttt{provable}.

We begin by establishing completeness.

\begin{lemma}[Completeness]\label{lem:completeness-wkn-alg}
Let $h$ be an irredundant $\omega$-hypersequent.
If $h$ has a proof in 
$\UniHCalcAbsorb{\UniHFLewRinv_\omega}$, then it has an $\omega$-eager proof.
\end{lemma}
\begin{proof}
We prove the stronger statement:
if $h$ is provable
in 
$\UniHCalcAbsorb{\UniHFLewRinv_\omega}$, then 
$T_\omega(\mHyperA)$ has an $\omega$-eager proof of $h$
for any \lrinhypersequent{}
$\mHyperA \UniSymbDef \UniEnrichHyper{\IdxHypA}{\UniRelDep{\mHyperA}}$, where $\MkIdx{h}$ has support $h$.

    By induction on the height $\alpha$ of a proof
    $\pi$ of $h$ in $\UniHCalcAbsorb{\UniHFLewRinv_\omega}$.
    Let $\mHyperA \UniSymbDef \UniEnrichHyper{\IdxHypA}{\UniRelDep{\mHyperA}}$
    be an arbitrary
    \lrinhypersequent{}.

\textit{(Base case)} This is obvious,
as the root of $T_\omega(\mHyperA)$ itself is an $\omega$-eager proof of $h$, since $h$ is an instance of an axiom of 
$\UniHCalcAbsorb{\UniHFLewRinv_\omega}$.

\textit{(Induction step)} 
Assume that
    the last rule
    applied in
    $\pi$
    was 
    $\UniRuleSchemaA = (h|t_i)_i / h$.
Thus, each $h_i=(h|t_i)$
has a proof $\pi_i$ in
$\UniHCalcAbsorb{\UniHFLewRinv_\omega}$
of height $\alpha_i < \alpha$.
By cases:
\begin{enumerate}
    \item $t_i$ occurs in $h$ for some $1 \leq i \leq k$:
    Say $h = h'|t_i$,
    and thus $h_i = h'|t_i|t_i$.
    Since $h_i$ has 
    a proof of height $< \alpha$, by
    hp-admissibility of $\UniEC$
    we know that $h$ itself
    has a proof of height $< \alpha$. 
    Because $h$ is irredundant by assumption, by IH $T_\omega(\mHyperA)$ has an $\omega$-eager proof
    of $h$.
\end{enumerate}

We know that each $(h\VL t_i)$ is irredundant. Thus, by IH,
any \lrinhypersequent{} 
$\mHyperA_i$ based on $\IdxHypA_i$
has an $\omega$-eager proof
of $h_i$.
Define the instantiation context $\mathfrak{I} \UniSymbDef \langle \UniRuleSchemaA,\mHyperA \rangle$.
We now consider the other cases.

\begin{enumerate}[leftmargin=*]
    \setcounter{enumi}{1}
    \item 
    For each $1 \leq i \leq k$, 
    it is not the case that
    $\MkIdx{t}_i$ has an $\omega$-partner in $\mathfrak{I}$:
    Then, since $t_i$
    does not occur in $h$,
    in the first step of the  construction of $T_\omega(\MkIdx{h})$,
    an \lrinhypersequent{} $\mHyperA_i$ will be added as a child of the root of $T_\omega(\mHyperA)$, and no $\omega$-refinement takes place.
    Thus, by construction, $T_\omega(\mHyperA_i)$
    is a subtree of $T_\omega(h)$ for each $1 \leq i \leq k$.
    Hence, picking the $\omega$-eager proofs $\pi'_i$ of $h_i$
    obtained by IH,
    obtain an $\omega$-eager proof of $h$ as a rooted subtree of $T_\omega(\MkIdx{h})$.
    \item 
    Now, assume wlog
    that $\MkIdx{t}_1,\ldots,\MkIdx{t}_m$
    have $\omega$-partners
    in $\mathfrak{I}$.
    This means that they have
    gone through the 
    $\omega$-refinement
    process when building
    $T_\omega(\mHyperA)$,
    and
    \lrinhypersequent{}s
    $\mHyperA'_i$
    with support
    $h'_i = (h|t'_i)$ were added as children of the root of $T_\omega(\mHyperA)$ for all $1 \leq i \leq m$, where $h'_i$ is clearly an $\omega$-extension of $h_i$.
    Since $h_i$ is provable with height $\alpha_i < \alpha$,
    due to Lemma~\ref{lem:hp-admin-omega-function}
    we have that $h|t'_i$
    has a proof with height $\leq \alpha_i < \alpha$.
    We now have two cases.

    If some $t'_i$ appears in $h$,
    then $h = h'|t'_i$
    and $h|t'_i = h'|t'_i|t'_i$. By hp-admissibility of $\UniEC$,
    $h$ has a proof of height $< \alpha$
    and thus by IH it has an $\omega$-eager proof
    in $T_\omega(\mHyperA)$.

    If no $t'_i$ appears in $h$, we know that
     $h'_i$ has an 
    $\omega$-eager proof
    in $T_\omega(\mHyperA_i')$
    by IH.
    Since 
    $T_\omega(\mHyperA_i')$
    is an immediate subtree of $T_{\omega}(\mHyperA)$,
    and since the IH applies to the other premises
    $(h|t_j)$, $j \geq m$,
    and $T_\omega(\mHyperA_j)$
    are also immediate subtrees of $T_\omega(\mHyperA)$ since they have not gone through $\omega$-refinement, we are done.\qedhere
\end{enumerate}
\end{proof}

Given an $\omega$-hypersequent~$h=(\AntOmSetA_0,\AntVecA_0)\allowbreak\Ra b_0 \VL \ldots \VL (\AntOmSetA_l,\AntVecA_l)\Ra b_l$ and $K\in\mathbb{N}$, let $h[K]$ denote the $\omega$-free $\omega$-hypersequent
$(\varnothing; \AntVecA_0[\AntOmSetA_0\allowbreak\mapsto K])\allowbreak\Ra b_0\VL\ldots\VL\allowbreak (\varnothing;\AntVecA_l[\AntOmSetA_l\mapsto K])\Ra b_l$. Informally speaking, every formula in the $\omega$-set has been assigned the multiplicity~$K$. 

\begin{figure*}
\centering
\begin{scriptsize}
\AxiomC{$h_0 \VL (\varnothing;\AntVecA_0) \Ra b \VL h_1 \VL (\varnothing;\AntVecB[\kappa\mapsto K_1]) \Ra b$}
\RightLabel{$\UniEW$,($\UniWProp$)}
\UnaryInfC{$h_0 \VL (\varnothing;\AntVecA_0) \Ra b \VL \ldots \VL h_1 \VL (\varnothing;\AntVecA_{K_1}) \Ra b$}
\AxiomC{$\{ h_0 \VL (\varnothing;\AntVecA_0) \Ra b \VL h_1\VL t_i\}$}
\RightLabel{$\UniEW$,($\UniWProp$)}
\UnaryInfC{$\{ h_0 \VL (\varnothing;\AntVecA_0) \Ra b \VL \ldots \VL h_1 \VL t_i'\}$}
\BinaryInfC{$h_0 \VL (\varnothing;\AntVecA_0) \Ra b \VL \ldots \VL h_1 \VL (\varnothing;\AntVecA_{k-1}) \Ra b \VL h_1$}
\noLine
\UnaryInfC{$\cdots$}
\noLine
\UnaryInfC{$h_0 \VL (\varnothing;\AntVecA_0) \Ra b \VL h_1 \VL (\varnothing;\AntVecA_1) \Ra b \VL h_1 \VL (\varnothing;\AntVecA_2)\Ra b$}
\AxiomC{$\{ h_0 \VL (\varnothing;\AntVecA_0) \Ra b \VL h_1\VL t_i\}$}
\RightLabel{$\UniEW$, $(\UniWProp)$}
\UnaryInfC{$\{ h_0 \VL (\varnothing;\AntVecA_0) \Ra b \VL h_1 \VL (\varnothing;\AntVecA_1) \Ra b \VL h_1 \VL t_i'\}$}
\BinaryInfC{$h_0 \VL (\varnothing;\AntVecA_0) \Ra b \VL h_1 \VL (\varnothing;\AntVecA_1) \Ra b \VL h_1$}
\noLine
\UnaryInfC{$\Pi_1$}
\noLine
\UnaryInfC{$h_0 \VL (\varnothing;\AntVecA_0) \Ra b \VL h_1 \VL (\varnothing;\AntVecA_1) \Ra b$}    
\AxiomC{\hspace{-4cm}$\{ h_0 \VL (\varnothing;\AntVecA_0) \Ra b \VL h_1 
\VL t_i\}$}
\BinaryInfC{$h_0 \VL (\varnothing;\AntVecA_0) \Ra b \VL h_1$}
\noLine
\UnaryInfC{$\Pi_0$}
\noLine
\UnaryInfC{$h_0 \VL (\varnothing;\AntVecA_0) \Ra b$}
\DisplayProof
\end{scriptsize}
\caption{
$K_1$ iterations of the derivation in (\ref{gadget-derivation}). The set of open leaves (non-axiomatic leaves) is unchanged.}
\label{gadget-iteration}
\end{figure*}

\begin{lemma}[Soundness]\label{lem:sound-eager}
If an $\omega$-free irredundant $\omega$-hypersequent has an $\omega$-eager proof,
then $h$ is provable in 
$\UniHCalcAbsorb{\UniHFLewRinv_\omega}$.
\end{lemma}
\begin{proof}
Let $h = \UniSequent{(\varnothing;\AntVecA)}{b}$
and let $\MkIdx{\pi}$ be an $\omega$-eager proof of $h$
of height $\alpha$, i.e.,
a subtree of $T_\omega(\mHyperA^{\circ})$
rooted at $\mHyperA^{\circ}$.
Obtain the tree $\pi$
by deleting all the indices of all components appearing
in $\MkIdx{\pi}$ (i.e., $\pi$ is a tree of $\omega$-hypersequents).
We must transform
$\pi$ into a proof in 
$\UniHCalcAbsorb{\UniHFLewRinv_\omega}$ (free of $\omega$-coordinates).
For that, it suffices to give a proof of $h$
containing no $\omega$-introduction rules; recall that the latter are those introduced as the productive output of $\omega$-refinement in the construction of $T_\omega(\mHyperA^{\circ})$.
Our transformation has two stages.

(\underline{Stage one}) We transform~$\pi$ by replacing, in every component (in every $\omega$-hypersequent occurring in the proof), each element of its $\omega$-set by a specified finite multiplicity. This multiplicity is obtained by a leaf-to-root induction on $\pi$ (equivalently, by an induction on the distance of a node in the tree from the leaves).

\textit{(Base case)} 
The node must be one of the following, where
we denote the logical constants by $\fff$ and $\ttt$ (since $0$ and $1$ are overloaded):
\begin{center}
\begin{small}
\begin{tabular}{ccc}
$g \VL \UniSequent{(\AntOmSetA;\mathbf{e}_b)}{b}$
&$g \VL \UniSequent{(\AntOmSetA\cup\{b\};\mathbf{0})}{b}$
&$g \VL \UniSequent{(\AntOmSetA;\mathbf{e}_{\fff})}{}$
\\[0.1em]
\multicolumn{3}{c}{
\begin{tabular}{cc}
$g \VL \UniSequent{(\AntOmSetA \cup \{\fff\}; \mathbf{0})}{}$
&$g \VL \UniSequent{(\AntOmSetA;\mathbf{0})}{\ttt}$
\end{tabular}
}
\end{tabular}
\end{small}
\end{center}

Map every $\omega$-set element to multiplicity~$1$ and prove using $(\UniWProp)$:
\begin{center}
\begin{small}
\begin{tabular}{ccc}
$g[1] \VL \UniSequent{(\varnothing;\mathbf e_b[\AntOmSetA\mapsto 1])}{b}$
&$g[1] \VL \UniSequent{(\varnothing;\mathbf{0}[\AntOmSetA\cup\{b\}\mapsto 1])}{b}$
\\[0.1em]
$g[1] \VL \UniSequent{(\varnothing;\mathbf{e}_{\fff}[\Omega \mapsto 1])}{}$
&$g[1] \VL \UniSequent{(\varnothing;\mathbf 0[\AntOmSetA \cup \{\fff\}\mapsto 1])}{}$
\\[0.1em]
\multicolumn{2}{c}{
$g[1] \VL \UniSequent{(\varnothing;\mathbf{0}[\AntOmSetA\mapsto 1])}{\ttt}$
}
\end{tabular}
\end{small}
\end{center}


\textit{(Inductive step)} 
Consider an internal node of the tree with label $g$, and assume that $\UniRuleSchemaA$
was applied backwards in the construction of $T_\omega(\MkIdx{h})$. 
Let~$\pi_i$ denote the subtree at the $i$-th premise, say with root labelled by $h_i$.
By inductive construction, each $\pi_i$ is transformed to $\pi_i'$, where the latter
has root $h'_i = h_i[K_i]$
for some $K_i \in \mathbb{N}$
(i.e., every $\omega$-set in $h_i$ is replaced by a finite multiplicity $K_i$). By cases:

\begin{enumerate}[leftmargin=*]
    \item $\UniRuleSchemaA$ is an instance of $\UniHCalcAbsorb{\UniHFLewRinv_\omega}$ with conclusion~$g$: 
    we identify a derivation
    of
    $\UniHCalcAbsorb{\UniHFLewRinv_\omega}$
    whose every non-axiom leaf is a root of some $\pi'_i$, and whose conclusion is $g[K_0]$, for a suitable
    $K_0\geq \max_i K_i$. Specifically, we show that there is an instance $\UniRuleSchemaA'$ of the same rule schema with conclusion $g[K_0]$ and premises that can be weakened upward to reach the roots of $\pi_i'$. I.e.,
\begin{center}
\begin{small}
    \AxiomC{$h_1[K_1]$}
\RightLabel{$(\UniWProp)^*$}
    \UnaryInfC{$g_1[K_0]$}
    \AxiomC{$\cdots$}
    \AxiomC{$h_m[K_m]$}
\RightLabel{$(\UniWProp)^*$}
    \UnaryInfC{$g_m[K_0]$}
    \RightLabel{$r'$}
    \TrinaryInfC{$g[K_0]$}
    \DisplayProof
\end{small}    
\end{center}
At the root of $\pi_i'$, for each component, we need $K_i$ copies of every formula that was formerly in its $\omega$-set. For components instantiating the hypersequent-variable in the rule schema of $r$, this is easy to fulfill; use the same instantiation. Meanwhile, every active component $c_a$ in the rule schema has multiset metavariables $M_1,\ldots,M_{l(a)}$, and the latter occur in some subset $\{p_1,\ldots,p_{n(a)}\}$ of principal components. The $\omega$-set of $c_a$ was defined (Def.~\ref{def:omega-struct-rules}) as the union of the $\omega$-sets of $\{p_1,\ldots,p_{n(a)}\}$. Hence, in the instantiation of each $M_i$, if we include $K$ (for some $K \in \mathbb{N}$) copies of each formula from the $\omega$-set of the principal component in which it occurs, then every formula that was formerly in the $\omega$-set of $c_a$ now has multiplicity $\geq K$ (excess copies may occur when an active component contains multiple multiset metavariables; these can be removed by weakening upwards to reach $K_i$). If $\UniRuleSchemaA$ is a logical rule, then we need an additional copy in the principal component, to serve (potentially) as principal formula. Recalling $\UniFmlaMultFixed$ (Def.~\ref{def:acn-fm}), replace every element in the $\omega$-set of the conclusion of $\UniRuleSchemaA$ by $K_0 = (\UniFmlaMultFixed \cdot \max_i K_i)+1$.

    \item $r$ is an $\omega$-introduction rule $\omega(r)$ that has a single productive premise (the argument extends by induction to the general case of multiple productive premises).  For spacing reasons, here and elsewhere we may present the premises stacked vertically.
\begin{center}
\AxiomC{$h_0 \;\VL\; (\AntOmSetA;\AntVecA) \Ra b \;\VL\; h_1 \;\VL\; (\AntOmSetA\cup\kappa; \AntVecB[\kappa\mapsto 0]) \Ra b$}
\noLine
\UnaryInfC{$\{ h_0 \VL (\AntOmSetA;\AntVecA) \Ra b \VL h_1\VL t_i\}_i$}
\RightLabel{$\omega(r)$}
\UnaryInfC{$h_0 \;\VL\; (\AntOmSetA;\AntVecA) \Ra b \;\VL\; h_1$}
\DisplayProof
\end{center}
Setting $K_0=\max_i K_i$, and replacing each~$\pi_i$ with $\pi_i'$ obtained from the inductive construction, write the following.
\begin{center}
\footnotesize
\AxiomC{$h_0[K_1] \;\VL\; (\varnothing;\AntVecA[\AntOmSetA\mapsto K_1]) \Ra b \;\VL\; h_1[K_1] \;\VL\; (\varnothing; \AntVecB[\AntOmSetA\cup\kappa\mapsto K_1]) \Ra b$}
\noLine
\UnaryInfC{$\{ (h_0 \VL (\AntOmSetA;\AntVecA) \Ra b \VL h_1\VL t_i)[K_i]\}_i$}
\RightLabel{$p\omega(r)$}
\UnaryInfC{$h_0[K_0] \;\VL\; (\varnothing;\AntVecA[\AntOmSetA\mapsto K_0]) \Ra b \;\VL\; h_1[K_0]$}
\DisplayProof
\end{center}
Unlike before, we do not obtain a rule instance of $\UniHCalcAbsorb{\UniHFLewRinv_\omega}$: 
the coordinates~$\kappa$ have some finite multiplicities in $(\varnothing;\AntVecA[\AntOmSetA\mapsto K])$ (conclusion), but they have multiplicity~$K_1$ in $\AntVecB[\AntOmSetA\cup\kappa\mapsto K_1]$ (premise).
Call the above a \emph{pseudo-$\omega$ rule instance} $p\omega(r)$, and call $\{K_i\}_{i\in I}$ its \emph{associated values}.
We will replace these with a derivation in the subsequent stage. For now, proceeding in this way, we ultimately reach a proof $\pi'$ of $(\varnothing, \AntVecA)\Ra b$ using rules of $\UniHCalcAbsorb{\UniHFLewRinv_\omega}$ + pseudo-$\omega$-rule instances. In particular, $\pi'$ does not contain $\omega$-introduction rules.
\end{enumerate}

(\underline{Stage two})
Now, let us take $\pi'$
obtained above
and
eliminate pseudo-$\omega$-rule instances, obtaining a proof of $h$ in $\UniHCalcAbsorb{\UniHFLewRinv_\omega}$.

We proceed by induction considering a slightly more general family of proofs.
Let~$P$ be a proof of an $\omega$-free hypersequent in $\UniHCalcAbsorb{\UniHFLewRinv_\omega}$ + pseudo-$\omega$-rule instances 
that satisfy the $\omega$-refinement condition, where each $\omega$-hypersequent in $P$ implicitly holds a (key ancestor, active component) relation on its components that is defined as in the $\omega$-proof search tree in Def~\ref{def:omega-proof-search}.

Associate a branch in $P$ with the number of pseudo-$\omega$-rule instance conclusions on it.
Define \emph{maximum nesting~$\rho(P)$} as the multiset of such numbers from all branches (i.e., every  branch of~$P$ contributes a single number). 
Induction on the standard multiset ordering on~$\rho$ (replacing a number in a multiset by finitely many strictly smaller numbers is taken to yield a strictly smaller multiset).

\textit{(Base case)}. If 
$\rho$ is empty, then $P$ is already a proof in $\UniHCalcAbsorb{\UniHFLewRinv_\omega}$.

\textit{(Inductive step)} Suppose that $\rho(P)$ is non-empty
and consider a lowermost instance of a pseudo-$\omega$-rule instance in $P$.
Such a rule instance appears as in the picture below. Here, $K_i$ is the $i$-th associated value of the pseudo-$\omega$-rule instance. Note that the $\omega$-set of every component in every $\omega$-hypersequent along the path from the root of $P$ to the conclusion of the chosen $p\omega(r)$ rule is empty, by choice of lowermost $p\omega(r)$ and assumption that the is $\omega$-free.
\begin{equation}
\label{gadget-derivation}
\text{
\AxiomC{$h_0 \;\VL\; (\varnothing;\AntVecA) \Ra b \;\VL\; h_1 \;\VL\; (\varnothing; \AntVecB[\kappa\mapsto K_1]) \Ra b$}
\noLine
\UnaryInfC{$\{ h_0 \VL (\varnothing;\AntVecA) \Ra b \VL h_1\VL t_i[\kappa\mapsto K_i]\}_i$}
\RightLabel{$p\omega(r)$}
\UnaryInfC{$h_0 \;\VL\; (\varnothing;\AntVecA) \Ra b \;\VL\; h_1$}
\noLine
\UnaryInfC{$\cdots$}
\noLine
\UnaryInfC{$h_0 \;\VL\; (\varnothing;\AntVecA) \Ra b$}
\DisplayProof
}
\end{equation}
Reading the above as a derivation where the only open leaves are those shown in the picture, replace its occurrence in $P$ by the derivation in Fig~\ref{gadget-iteration}. Here, $\AntVecA_0,\ldots,\AntVecA_{K_1}$ is the sequence defined by $\AntVecA_0:=\AntVecA$ and $\AntVecA_{i+1}=\AntVecA+(i+1)(\AntVecB-\AntVecA)$ (recall that $\AntVecA < \AntVecB$ by the condition in $\omega$-refinement). Call this new tree $P^*$.

First, let us explain why this derivation has the same set of open leaves as before.
Consider the first premise of $p\omega(r)$. In the (key ancestor, active component) dependency graph on components, we have that the featured $(\varnothing;\AntVecA_0) \Ra b$
is an ancestor of the active component $(\varnothing;\AntVecA_1) \Ra b$, by the $\omega$-refinement condition. Moreover, looking at the rule schemas underlying the rule instances from the introduction of $(\varnothing;\AntVecA_0) \Ra b$ to $(\varnothing;\AntVecA_1) \Ra b$, there is a common multiset metavariable between each (key ancestor, active component) pair---ensuring this property was a major motivation in the definition of key ancestor. Now, $\Pi_{i}$ ($i>1$) is like $\Pi_0$ except that the excess $\AntVecA_{i}- \AntVecA_0$ is included as part of the instantiation of these multiset metavariables. Although this excess will now appear in every active component of a rule instance within $\Pi$ that contains such a multiset metavariable---i.e., the unproductive premises of the $\omega$-introduction rule, or the premise of a branching rule inside a $\Pi_i$ that leads towards an axiom---all unwanted copies can be deleted by weakening upwards. Hence  $\Pi_i$ is indeed a derivation.

Next, we claim that the nesting depth~$\rho^*$ of $P^*$ is strictly less than $\rho$ in the multiset ordering. 
Any branch~$B$ in $P$ either passes our chosen $p\omega(r)$, or it does not.
Let $\rho_1$ be the multiset of nesting depths from branches of the former type, and $\rho_2$ from the latter.
Thus, $\rho=\rho_1\uplus \rho_2$, where $\uplus$ is multiset union.
Also, any branch in $P^*$ either has the form $\sigma B$ where $\sigma$ is the path from the root to an open leaf wrt the deduction ($B$ is the remainder of the path), or it does not pass any open leaf wrt the deduction (this means that it branches inside some $\Pi_i$ and heads towards an axiom).
Then, $\rho^* = \rho_1^* \uplus \rho_2$, where $\rho_1^*$ is the multiset of nesting depths from branches of the former type, and $\rho_2$ is from the latter. Since $\sigma$ does not contain any pseudo-$\omega$-rule instance, $\rho_1 = 1 +\rho_1^*$ (where $+$ is the obvious pointwise addition operation). Hence, the claim is established.

To establish the inductive case, it suffices now to apply IH to $P^*$.
Finally, applying this result to $\pi'$ yields the desired proof.
\end{proof}

\begin{theorem}[Correctness of $\UniPSWknAlgnoinput$]\label{the:wkn-correctness}
For any irredundant hypersequent
$h_0$, 
$\UniPSWknAlg{h_0}$
outputs {\normalfont\texttt{yes}}
iff
$h_0$ is provable in
$\UniHCalcAbsorb{\UniHFLewRinv_\omega}$.
\end{theorem}



From 
Lem~\ref{lem:original-to-invert},
Lem~\ref{lem:original-calc-omega-calc}, Lem~\ref{lem:omega-calc-omega-refined}, 
Thm~\ref{the:wkn-correctness}
and Thm~\ref{the:complexity-wkn}:

\begin{theorem}[main theorem II]
Provability in any extension of
$\UniFLeExtLogic{\UniWProp}$ 
by finitely many $\mathcal{P}_3$-axioms
(i.e, provability in any analytic structural extension of $\UniFLeExtHCalc{\UniWProp}$)
has Ackermannian complexity.
\end{theorem}

\begin{corollary}
    $\mathbf{MTL}$
    has Ackermannian complexity.
\end{corollary}



\printbibliography

\end{document}